\documentclass[10pt, twocolumn]{IEEEtran}

\usepackage{psfrag}
\usepackage{amssymb}
\usepackage{amsmath}
\usepackage{pifont}
\usepackage{cite}
\usepackage{graphics}
\usepackage{graphicx}
\usepackage{epsfig}
\usepackage{subfigure}
\usepackage{url}
\usepackage{amscd}
\usepackage{threeparttable}
\usepackage[colorlinks, citecolor=blue,linkcolor=blue]{hyperref}
\usepackage{enumerate}

\newtheorem{theorem}{Theorem}

\newtheorem{lemma}{Lemma}
\newtheorem{definition}{Definition}
\newtheorem{proposition}{Proposition}

\newtheorem{example}{Example}

\begin{document}

\title{Recovery of Sparse Signals Using Multiple Orthogonal Least Squares}

\author{\IEEEauthorblockN{Jian Wang and Ping Li} \\
\IEEEauthorblockA{Department of Statistics and Biostatistics,
Department of Computer Science \\Rutgers, The State University of New Jersey \\
Piscataway, New Jersey 08854, USA \\
E-mail: \{jwang,pingli\}@stat.rutgers.edu}
}

\maketitle

\begin{abstract}
We study the problem of recovering sparse signals from compressed linear
measurements. This problem, often referred to as sparse recovery or sparse reconstruction, has generated a great deal of interest in
recent years. To recover the sparse signals, we propose a new method
called multiple orthogonal least squares (MOLS), which extends the well-known orthogonal least squares (OLS) algorithm
by allowing multiple $L$ indices to be chosen per iteration. Owing to
inclusion of multiple support indices in each selection, the MOLS
algorithm converges in much fewer iterations and improves the
computational efficiency over the conventional OLS algorithm. Theoretical analysis
shows that MOLS ($L > 1$) performs exact recovery of all $K$-sparse signals
within $K$ iterations if the measurement matrix satisfies the
restricted isometry property (RIP) with isometry constant
$\delta_{LK} < \frac{\sqrt{L}}{\sqrt{K} + 2 \sqrt{L}}.$
The recovery performance of MOLS in the noisy scenario is also studied. It is shown that stable recovery of sparse signals can be achieved with the MOLS algorithm when the signal-to-noise ratio (SNR) scales linearly with the sparsity level of input signals.

\end{abstract}


{\IEEEpeerreviewmaketitle}

\section{Introduction}\label{sec:intro}

In recent years, sparse recovery has attracted much attention in
applied mathematics, electrical engineering, and statistics~\cite{donoho1989uncertainty,donoho2006compressed,candes2006near,candes2006robust}.
The main task of sparse recovery is to recover a high dimensional $K$-sparse vector $\mathbf{x} \in \mathcal{R}^{n}$ ($\| \mathbf{x} \|_{0} \leq K \ll n$) from a small number of linear measurements
\begin{equation} \label{eq:1}
  \mathbf{y} = \mathbf{\Phi x},
\end{equation}
where $\mathbf{\Phi} \in \mathcal{R}^{m \times n}$ ($m < n$) is often called the measurement matrix. Although the system is
underdetermined, owing to the signal sparsity, $\mathbf{x}$ can be accurately recovered from the measurements $\mathbf{y}$ by solving an $\ell_0$-minimization problem:
\begin{equation} \label{eq:2nb}
\min_\mathbf{x} \| \mathbf{x} \|_0~~\text{subject to}~~~~\mathbf{y}
= \mathbf{\Phi x}.
\end{equation}
This method, however, is known to be intractable due to the
combinatorial search involved and therefore impractical for
realistic applications. Thus, much attention has focused on developing efficient algorithms for recovering the sparse signal. In general, the algorithms can be classified into two major categories: those using convex optimization techniques~\cite{donoho1989uncertainty,chen2001atomic,donoho2006compressed,candes2006near,candes2006robust} and those based on greedy searching principles~\cite{pati1993orthogonal,mallat1993matching,chen1989orthogonal,donoho2006sparse,needell2010signal,wang2012Generalized,needell2009cosamp,foucart2011hard,dai2009subspace
}. Other algorithms relying on nonconvex methods have also been proposed~\cite{chartrand2007exact,chartrand2008iteratively,foucart2009sparsest,daubechies2010iteratively,chen2014convergence}.
The optimization-based approaches replace the nonconvex $\ell_0$-norm with its convex surrogate  $\ell_1$-norm, translating
the combinatorial hard search into a computationally tractable
problem:
\begin{equation}
\min_\mathbf{x} \| \mathbf{x} \|_1~~\text{subject to}~~~~\mathbf{y}
= \mathbf{\Phi x}.
\end{equation}
This algorithm is known as basis pursuit (BP)~\cite{chen2001atomic}. It has been revealed that under appropriate constraints on the measurement matrix, BP yields  exact recovery of the sparse signal.

The second category of approaches for sparse recovery are greedy
algorithms, in which signal support is iteratively identified
according to various greedy principles. Due to their computational simplicity and competitive performance, greedy algorithms have
gained considerable popularity in practical applications. Representative algorithms include
matching pursuit (MP)~\cite{mallat1993matching}, orthogonal matching
pursuit (OMP)~\cite{pati1993orthogonal,tropp2004greed,tropp2007signal,davenport2010analysis,zhang2011sparse,mo2012remarks,wang2012Recovery,wen2013improved,wang2015support}
and orthogonal least squares (OLS)~\cite{chen1989orthogonal,rebollo2002optimized,foucart2013stability,soussen2013joint}.
Both OMP and OLS identify the support of the underlying sparse signal
by adding one index at a time, and estimate the sparse coefficients
over the enlarged support. The main difference between OMP and OLS
lies in the greedy rule of updating the support at each iteration. While OMP finds a
column that is most strongly correlated with the signal residual,
OLS seeks to maximally reduce the power of the current residual with an enlarged
support set. It has been shown that OLS has better convergence
property but is computationally more expensive than the OMP algorithm~\cite{soussen2013joint}.


In this paper, with the aim of improving the recovery accuracy and also reducing the computational cost of
OLS, we propose a new method called multiple orthogonal least
squares (MOLS), which can be viewed as an extension of the OLS algorithm in that multiple indices are allowed to be chosen at a time. Our
method is inspired by that those sub-optimal candidates in
each of the OLS identification are likely to be reliable and could be utilized to better reduce the power of signal residual for each iteration, thereby accelerating the convergence of the algorithm.
The main steps of the MOLS algorithm are specified in Table
\ref{tab:mols}. Owing to selection of multiple ``good''
candidates in each time, MOLS converges in much fewer iterations and improves the
computational efficiency over the
conventional OLS algorithm.

\setlength{\arrayrulewidth}{1.5pt}
\begin{table}[t]
\begin{center}
\caption{The MOLS Algorithm} \label{tab:mols} 
\begin{tabular}{@{}ll}
\hline \\ \vspace{-13pt} \\
\textbf{Input}       &measurement matrix $\mathbf{\Phi} \in \mathcal{R}^{m \times n}$,\\
                     &measurements vector $\mathbf{y} \in \mathcal{R}^{m}$, \\
                     &sparsity level $K$, \\
                     &and selection parameter $L \leq \min \{K,  \frac{m}{K} \}$. \\
\textbf{Initialize}  &iteration count $k = 0$, \\
                     &estimated support $\mathcal{T}^{0} = \emptyset$, \\
                     &and residual vector $\mathbf{r}^{0} = \mathbf{y}$.
                     \\
\textbf{While}       &($\|\mathbf{r}^k\|_2 \geq \epsilon$ and $k < K$) or $Lk <K$, \textbf{do}\\
                     & $k = k + 1$. \\
                     & Identify \hspace{1.15mm}$\mathcal{S}^{k} = \underset{\mathcal{S} : | \mathcal{S} | =L}{\arg \min} \sum_{i \in \mathcal{S}} \| \mathbf{P}^{\bot}_{\mathcal{T}^{k - 1} \cup \{i\}} \mathbf{y} \|_{2}^{2}$. \\
                     & Enlarge ~$\mathcal{T}^{k} = \mathcal{T}^{k - 1} \cup \mathcal{S}^{k}$. \\
                     & Estimate \hspace{0.41mm}$\mathbf{x}^{k} = \underset{\mathbf{u}:\textit{supp}(\mathbf{u}) = \mathcal{T}^k}{\arg \min} \|\mathbf{y}-\mathbf{\Phi} \mathbf{u}\|_2$. \\
                     & Update \hspace{2.5mm}$\mathbf{r}^{k} = \mathbf{y} - \mathbf{\Phi} \mathbf{x}^{k}$. \\
\textbf{End}         & \\
\textbf{Output}      & the estimated support $\hat{\mathcal{T}} =
\underset{\mathcal{S}:|\mathcal{S}| = K}{\arg \min} \|\mathbf{x}^k -
\mathbf{x}^k_\mathcal{S}\|_2$ and the\\
& estimated signal $\hat{\mathbf{x}}$ satisfying $\hat{\mathbf{x}}_{\Omega \setminus \hat{\mathcal{T}}} = \mathbf{0}$ and $\hat{\mathbf{x}}_{\hat{\mathcal{T}}} = \mathbf{\Phi}_{\hat{\mathcal{T}}}^\dag \mathbf{y}$.  \\
\vspace{-5pt} \\
\hline
\end{tabular}
\end{center}
\end{table}

\setlength{\arrayrulewidth}{1pt}

Greedy methods with a similar flavor to MOLS in adding multiple
indices per iteration include stagewise OMP (StOMP)~\cite{donoho2006sparse}, regularized OMP (ROMP)~\cite{needell2010signal}, and generalized OMP (gOMP)~\cite{wang2012Generalized} (also known as orthogonal super greedy
algorithm (OSGA)~\cite{liu2012orthogonal}), etc. These algorithms
identify candidates at each iteration according to correlations
between columns of the measurement matrix and the residual vector.
Specifically, StOMP picks indices whose magnitudes of correlation
exceed a deliberately designed threshold. ROMP first chooses a set
of $K$ indices with strongest correlations and then narrows down the
candidates to a subset based on a predefined regularization rule.
The gOMP algorithm finds a fixed number of indices with strongest correlations in
each selection.
Other greedy methods adopting a different strategy of adding as well
as pruning indices from the list include compressive sampling
matching pursuit (CoSaMP)~\cite{needell2009cosamp} and subspace
pursuit (SP)~\cite{dai2009subspace} and hard thresholding pursuit
(HTP)~\cite{foucart2011hard}, etc.

The contributions of this paper are summarized as follows.
\begin{enumerate}[i)]
  \item We propose a new algorithm, referred to as MOLS, for solving sparse recovery problems. We analyze the MOLS algorithm using the restricted isometry property (RIP) introduced in the compressed sensing (CS) theory~\cite{candes2005decoding} (see Definition~\ref{def:rip} below). Our analysis shows that MOLS ($L>1$) exactly recovers any $K$-sparse signal within $K$ iterations if the  measurement matrix $\mathbf{\Phi}$ obeys the RIP
  with isometry constant
  \begin{equation}
    \delta_{LK} < \frac{\sqrt{L}}{\sqrt{K} + 2 \sqrt{L}}. \label{eq:good}
  \end{equation}
  For the special case when $L=1$, MOLS reduces to the conventional OLS algorithm. We establish the condition for the exact sparse recovery with OLS as
  \begin{equation}
    \label{eq:ols} \delta_{K+1} < \frac{1}{\sqrt{K} + 2}.
  \end{equation}
This condition is nearly sharp in the sense that, even with a slight relaxation (e.g., relaxing to
  $\delta_{K+1} = \frac{1}{\sqrt{K}}$), the exact recovery with OLS may not be guaranteed.

%

  \item We analyze recovery performance of MOLS in the presence of noise. Our result demonstrates that stable recovery of sparse signals can be achieved with MOLS when the signal-to-noise ratio (SNR) scales linearly with the sparsity level of input signals. In particular, for the case of OLS (i.e., when $L = 1$), we show that the scaling law of the SNR is necessary for exact support recovery of sparse signals.


\end{enumerate}

The rest of this paper is organized as follows: In Section~\ref{sec:II}, we
introduce notations, definitions, and lemmas that will be used in this paper. In
Section~\ref{sec:III}, we give a useful observation regarding the identification step of MOLS. In Section
\ref{sec:IV} and~\ref{sec:V}, we analyze the theoretical performance of MOLS in recovering sparse signals. In Section~\ref{sec:VI}, we study the empirical performance of the MOLS algorithm.
Concluding remarks are given in Section~\ref{sec:VII}.

\section{Preliminaries} \label{sec:II}

\subsection{Notations}

We first briefly summarize notations used in this paper. Let $\Omega
= \{1,2, \cdots, n\}$ and let $\mathcal{T}=
\textit{supp}(\mathbf{x}) = \{i|i \in \Omega,x_{i} \neq 0\}$ denote
the support of vector $\mathbf{x}$. For $\mathcal{S} \subseteq
\Omega$, $|\mathcal{S}|$ is the cardinality of $\mathcal{S}$.
$\mathcal{T} \setminus \mathcal{S}$ is the set of all elements
contained in $\mathcal{T}$ but not in $\mathcal{S}$.
$\mathbf{x}_{\mathcal{S}} \in \mathcal{R}^{|\mathcal{S}|}$ is the
restriction of the vector $\mathbf{x}$ to the elements with indices
in $\mathcal{S}$. $\mathbf{\Phi}_{\mathcal{S}} \in \mathcal{R}^{m
\times | \mathcal{S} |}$ is a submatrix of $\mathbf{\Phi}$ that only
contains columns indexed by $\mathcal{S}$. If
$\mathbf{\Phi}_{\mathcal{S}}$ is full column rank, then
$\mathbf{\Phi}_{\mathcal{S}}^{\dagger} = (
\mathbf{\Phi}'_{\mathcal{S}} \mathbf{\Phi}_{\mathcal{S}} )^{-1}
\mathbf{\Phi}'_{\mathcal{S}}$ is the pseudoinverse of
$\mathbf{\Phi}_{\mathcal{S}}$. $\text{span} (
\mathbf{\Phi}_{\mathcal{S}} )$ is the span of columns in
$\mathbf{\Phi}_{\mathcal{S}}$. $\mathbf{P}_{\mathcal{S}} =
\mathbf{\Phi}_{\mathcal{S}} \mathbf{\Phi}_{\mathcal{S}}^{\dagger}$
is the projection onto $\text{span} ( \mathbf{\Phi}_{\mathcal{S}}
)$. $\mathbf{P}_{\mathcal{S}}^{\bot} = \mathbf{I} -
\mathbf{P}_{\mathcal{S}}$ is the projection onto the orthogonal
complement of $\text{span} ( \mathbf{\Phi}_{\mathcal{S}} )$, where
$\mathbf{I}$ is the identity matrix. For mathematical convenience,
we assume that $\mathbf{\Phi}$ has unit $\ell_{2}$-norm columns
throughout the paper.\footnote{In~\cite{blumensath2007difference},
it has been shown that the behavior of OLS is unchanged whether the
columns of $\mathbf{\Phi}$ have unit $\ell_2$-norm or not. As MOLS
is a direct extension of the OLS algorithm, it can be verified that
the behavior of MOLS is also unchanged whether $\mathbf{\Phi}$ has
unit $\ell_2$-norm columns or not.}

\subsection{Definitions and Lemmas}

\begin{definition}[RIP~\cite{candes2005decoding}] \label{def:rip}
A measurement matrix $\mathbf{\Phi}$ is said to satisfy
the RIP of order $K$ if there exists a constant $\delta \in (0, 1)$
such that
\begin{equation}
  \label{eq:RIP}
  (1 - \delta) \| \mathbf{x} \|_{2}^2 \leq \|
  \mathbf{\Phi x} \|_{2}^2 \leq (1 + \delta) \| \mathbf{x}
  \|_{2}^2
\end{equation}
for all $K$-sparse vectors $\mathbf{x}$. In particular, the minimum
of all constants $\delta$ satisfying (\ref{eq:RIP}) is called the isometry
constant $\delta_{K}$.
\end{definition}

The following lemmas are useful for our analysis.

\begin{lemma}
  [Lemma 3 in {\cite{candes2005decoding}}]\label{lem:mono}If a measurement
  matrix satisfies the RIP of both orders $K_{1}$ and $K_{2}$ where $K_{1} \leq K_{2}$, then
  $\delta_{K_{1}} \leq \delta_{K_{2}}$. This
  property is often referred to as the monotonicity of the isometry constant.
\end{lemma}

\begin{lemma}
  [Consequences of RIP {\cite{needell2009cosamp,kwon2013multipath}}]\label{lem:rips} Let $\mathcal{S}
  \subseteq \Omega$. If $\delta_{| \mathcal{S} |} <1$ then for any vector $\mathbf{u} \in
  \mathcal{R}^{| \mathcal{S} |}$,
  \begin{eqnarray}
    &&(1- \delta_{| \mathcal{S} |} ) \left\| \mathbf{u} \right\|_{2} \leq \left\|
    \mathbf{\Phi}_{\mathcal{S}}'  \mathbf{\Phi}_{\mathcal{S}} \mathbf{u} \right\|_{2} \leq ( 1+
    \delta_{| \mathcal{S} |} ) \left\| \mathbf{u} \right\|_{2}, \nonumber\\
    &&~~~~~~\frac{\| \mathbf{u} \|_{2}}{1+ \delta_{| \mathcal{S} |}} \leq \| (
    \mathbf{\Phi}_{\mathcal{S}}'  \mathbf{\Phi}_{\mathcal{S}} )^{-1} \mathbf{u} \|_{2} \leq
    \frac{\| \mathbf{u} \|_{2}}{1- \delta_{| \mathcal{S} |}}. \nonumber
  \end{eqnarray}
\end{lemma}

\begin{lemma}
  [Lemma 2.1 in {\cite{candes2008restricted}}]\label{lem:correlationrip} Let
  $\mathcal{S}_{1},\mathcal{S}_{2} \subseteq \Omega$ and $\mathcal{S}_{1} \cap \mathcal{S}_{2} =
  \emptyset$. If $\delta_{|\mathcal{S}_{1} | + |\mathcal{S}_{2} |} <1$, then
  $\| \mathbf{\Phi}_{\mathcal{S}_{1}}'  \mathbf{\Phi v} \|_{2} \leq \delta_{|\mathcal{S}_{1} | +
     |\mathcal{S}_{2} |} \left\| \mathbf{v} \right\|_{2}$
  holds for any vector $\mathbf{v} \in \mathcal{R}^{n}$ supported on $\mathcal{S}_{2}$.
\end{lemma}

\vspace{1mm}

\begin{lemma}[Proposition 3.1 in~\cite{needell2009cosamp}] \label{lem:rip5}
Let $\mathcal{S} \subset \Omega$. If $\delta_{|\mathcal{S} |} < 1$, then for any vector $\mathbf{u} \in \mathcal{R}^m$,
$
\|\mathbf{\Phi}'_\mathcal{S} \mathbf{u}\|_2 \leq \sqrt{1 + \delta_{|\mathcal{S}|}} \|\mathbf{u}\|_2.$
\end{lemma}
\vspace{1mm}
\begin{lemma}[Lemma 5 in~\cite{cai2011orthogonal}] \label{lem:rip6} Let
  $\mathcal{S}_{1},\mathcal{S}_{2} \subseteq \Omega$. Then the minimum and maximum eigenvalues of $\mathbf{\Phi}'_{\mathcal{S}_1}  \mathbf{P}^\bot_{\mathcal{S}_2}  \mathbf{\Phi}_{\mathcal{S}_1  }$ satisfy
\begin{eqnarray}
 \lambda_{\min} (  \mathbf{\Phi}'_{\mathcal{S}_1}  \mathbf{P}^\bot_{\mathcal{S}_2}  \mathbf{\Phi}_{\mathcal{S}_1  })  &\geq& \lambda_{\min} (\mathbf{\Phi}'_{\mathcal{S}_1 \cup \mathcal{S}_2} \mathbf{\Phi}_{\mathcal{S}_1 \cup \mathcal{S}_2}), \nonumber \\
  \lambda_{\max} (  \mathbf{\Phi}'_{\mathcal{S}_1}  \mathbf{P}^\bot_{\mathcal{S}_2}  \mathbf{\Phi}_{\mathcal{S}_1  })  &\leq& \lambda_{\max} (\mathbf{\Phi}'_{\mathcal{S}_1 \cup \mathcal{S}_2} \mathbf{\Phi}_{\mathcal{S}_1 \cup \mathcal{S}_2}). \nonumber
\end{eqnarray}

\end{lemma}

\begin{lemma} \label{lem:rips2} Let $\mathcal{S}
  \subseteq \Omega$. If $\delta_{| \mathcal{S} |} <1$ then for any vector $\mathbf{u} \in
  \mathcal{R}^{| \mathcal{S} |}$,
\begin{equation}
      \frac{\left\| \mathbf{u} \right\|_{2}}{\sqrt{1+ \delta_{| \mathcal{S} |}}}  \leq \| (\mathbf{\Phi}_{\mathcal{S}}^\dag)' \mathbf{u} \|_{2} \leq
    \frac{\left\| \mathbf{u} \right\|_{2}}{\sqrt{1- \delta_{| \mathcal{S} |}}}.
    \label{eq:ab0}
  \end{equation}
\end{lemma}
The upper bound in \eqref{eq:ab0} has appeared in \cite[Proposition 3.1]{needell2009cosamp} and we give a proof for the upper bound in Appendix~\ref{app:new}.

\section{Observation} \label{sec:III}

Let us begin with an interesting and important observation regarding the identification
step of MOLS as shown in Table~\ref{tab:mols}. At the $(k+1)$-th
iteration ($k \geq 0$), MOLS adds to $\mathcal{T}^{k}$ a set of $L$
indices,
\begin{equation}
  \label{eq:golsrulej} \mathcal{S}^{k+1} = \arg \min_{\mathcal{S} : | \mathcal{S} | =L}
  \sum_{i \in \mathcal{S}} \| \mathbf{P}^{\bot}_{\mathcal{T}^{k} \cup \{i\}} \mathbf{y}
  \|_{2}^{2}.
\end{equation}
Intuitively, a straightforward implementation of
\eqref{eq:golsrulej} requires to sort all elements in $\{\|
\mathbf{P}^{\bot}_{\mathcal{T}^{k} \cup \{i\}} \mathbf{y} \|_{2}^{2}
\}_{i \in \Omega \setminus \mathcal{T}^{k}}$ and then find the
smallest $L$ ones (and their corresponding indices). This
implementation, however, is computationally expensive as it requires
to construct $n-Lk$ different orthogonal projections (i.e.,
$\mathbf{P}^{\bot}_{\mathcal{T}^{k} \cup \{i\}}$, $\forall i \in
\Omega \setminus \mathcal{T}^{k}$). Therefore, it is highly
desirable to find a cost-effective alternative to
\eqref{eq:golsrulej} for the identification step of MOLS.

Interestingly, the following proposition illustrates that \eqref{eq:golsrulej} can be substantially simplified.  It is inspired by the technical report of Blumensath and Davies~\cite{blumensath2007difference}, in which a geometric interpretation of OLS is given in terms of orthogonal projections.

\begin{proposition} \label{prop:p1}
At the $(k + 1)$-th iteration, the MOLS algorithm identifies a set of $L$ indices:
\begin{align}\label{eq:golsrule11111}
  \mathcal{S}^{k + 1} =&
   \arg \max_{\mathcal{S} : | \mathcal{S} | =L}   \sum_{i \in \mathcal{S}}
\frac{| \langle \phi_{i}, \mathbf{r}^k \rangle |}{\|
\mathbf{P}^{\bot}_{\mathcal{T}^{k}} \phi_{i} \|_{2}}\\ \label{eq:golsrule11110}
   =&\arg \max_{\mathcal{S} : | \mathcal{S} | =L}   \sum_{i \in \mathcal{S}}
\bigg| \bigg\langle \frac{ \mathbf{P}^{\bot}_{\mathcal{T}^{k}}
\phi_{i} }{\| \mathbf{P}^{\bot}_{\mathcal{T}^{k}} \phi_{i} \|_{2}},
\mathbf{r}^k \bigg\rangle \bigg|.
\end{align}
\end{proposition}

The proof of this proposition is given in Appendix~\ref{app:p1}. It is essentially identical to some analysis in~\cite{rebollo2002optimized} (which is particularly for OLS), but with extension to the case of selecting multiple indices per iteration (i.e., the MOLS case). This extension is important in that it not only enables a low-complexity implementation for MOLS, but also will play an key role in the
performance analysis of MOLS in Section~\ref{sec:IV} and~\ref{sec:V}.
We thus include the proof for completeness.

One can interpret from \eqref{eq:golsrule11111} that
to identify $\mathcal{S}^{k+1}$, it suffices to find the $L$ largest
values in $\big\{ \frac{| \langle \phi_{i}, \mathbf{r}^{k} \rangle
|}{\| \mathbf{P}^{\bot}_{\mathcal{T}^{k}} \phi_{i} \|_{2}}
\big\}_{i \in \Omega \setminus \mathcal{T}^{k}}$, which is
much simpler than \eqref{eq:golsrulej} as it involves only one
projection operator (i.e.,
$\mathbf{P}^{\bot}_{\mathcal{T}^{k}}$).
Indeed, by numerical experiments, we
have confirmed that the simplification offers massive reduction in the computational cost.


Following the arguments in~\cite{blumensath2007difference,soussen2013joint}, we give a geometric interpretation of the selection rule in
MOLS: the columns of measurement matrix are projected onto the
subspace that is orthogonal to the span of the active columns, and
the $L$ normalized projected columns that are best correlated with
the residual vector are selected.


\section{Exact Sparse Recovery with MOLS} \label{sec:IV}

\subsection{Main Results} \label{sec:mainresults}

In this section, we study the condition of MOLS for exact recovery of sparse signals. For convenience of  stating the results, we say that MOLS makes a success at an
iteration if it selects at least one correct index at the iteration.
Clearly if MOLS makes a success in each iteration, it will select all support indices within $K$ iterations. When all support indices are selected, MOLS can recover the sparse signal exactly.

\begin{theorem}
  \label{thm:atleast10}
  Let $\mathbf{x} \in \mathcal{R}^n$ be any $K$-sparse signal and let $\mathbf{\Phi} \in \mathcal{R}^{m \times n}$ be the measurement matrix. Also, let $L$ be the number of indices selected at each iteration of MOLS. Then if $\mathbf{\Phi}$ satisfies the RIP with
  \begin{equation}
    \left\{\begin{array}{ll}
      \delta_{LK} < \frac{\sqrt{L}}{\sqrt{K} + 2\sqrt{L}}, & L>1,\\
      \delta_{K+1} < \frac{1}{\sqrt{K} + 2}, & L=1,
    \end{array}\right. \label{eq:jjjjffffa}
  \end{equation}
  MOLS exactly recovers $\mathbf{x}$ from the
  measurements $\mathbf{y} = \mathbf{\Phi x}$ within $K$ iterations.
\end{theorem}

Note that when $L=1$, MOLS reduces to the conventional OLS algorithm. Theorem
\ref{thm:atleast10} suggests that under $\delta_{K+1} <
\frac{1}{\sqrt{K} +2}$, OLS can recover any $K$-sparse signal in exact $K$ iterations.
Similar results have also been established for the OMP algorithm.  In~\cite{mo2012remarks,wang2012Recovery}, it has been shown that $\delta_{K+1} < \frac{1}{\sqrt{K} + 1}$ is sufficient for OMP to exactly recover $K$-sparse signals. The condition is recently improved to $\delta_{K + 1} < \frac{\sqrt{4K + 1} - 1}{2K}$~\cite{chang2014improved}, by utilizing techniques developed in~\cite{chang2013achievable}. It is worth mentioning that there exist examples of measurement matrices satisfying $\delta_{K+1} = \frac{1}{\sqrt{K}}$ and $K$-sparse signals, for which OMP makes wrong selection at the first iteration and thus fails to recover the signals in $K$ iterations~\cite{mo2012remarks,wang2012Recovery}.
Since OLS coincides with OMP for the first iteration, those examples naturally apply to OLS, which therefore implies that $\delta_{K + 1} < \frac{1}{\sqrt{K}}$ is a necessary condition for the OLS algorithm. We would like to mention that the claim of $\delta_{K + 1} < \frac{1}{\sqrt{K}}$ being necessary for exact recovery with OLS has also been proved in~\cite[Lemma 1]{herzet2012exact}. Considering the fact that $\frac{1}{\sqrt{K} + 2}$ converges to $\frac{1}{\sqrt{K}}$ as $K$ goes large, the proposed condition $\delta_{K+1} < \frac{1}{\sqrt{K} + 2}$ is nearly sharp.

The proof of Theorem~\ref{thm:atleast10} follows along a similar
line as the proof in~\cite[Section III]{wang2012Generalized}, in which conditions
for exact recovery with gOMP were proved using mathematical induction, but with two key
distinctions. The first distinction lies in the way we lower bound
the correlations between correct columns and the current residual.
As will be seen in Proposition~\ref{lem:upperbound}, by employing an
improved analysis, we obtain a tighter bound for MOLS than the
corresponding result for gOMP~\cite[Lemma 3.7]{wang2012Generalized},
which consequently leads to better recovery conditions for the MOLS
algorithm. The second distinction is in how we obtain recovery conditions for the case of $L = 1$. While in~\cite[Theorem 3.11]{wang2012Generalized}
the condition for the first iteration of OMP directly
applies to the general iteration and hence becomes an overall condition for OMP (i.e., the condition for success of
the first iteration also guarantees the success of succeeding
iterations, see~\cite[Lemma 3.10]{wang2012Generalized}), such is not the case for the OLS algorithm due to the
difference in the identification rule.
This means that to obtain the overall condition for OLS, we need to consider the first iteration and the general iteration individually, which makes the underlying analysis for OLS more complex and also leads to a more restrictive condition than the OMP algorithm.

\subsection{Proof of Theorem~\ref{thm:atleast10}} \label{sec:proofoftheorem1}

The proof works by mathematical induction. We first establish a
condition that guarantees success of MOLS at the first iteration.
Then we assume that MOLS has been successful in previous $k$
iterations ($1 \leq k < K$). Under this assumption, we derive a
condition under which MOLS also makes a success at the $(k + 1)$-th
iteration. Finally, we combine these two conditions to obtain an
overall condition.

\subsubsection{Success of the first iteration}
From {\eqref{eq:golsrule11111}}, MOLS selects at the
  first iteration the index set
  \begin{eqnarray}
    \mathcal{T}^{1}
     &\hspace{-3mm} = & \hspace{-3mm} \arg   \max_{\mathcal{S} : | \mathcal{S} | =L}   \sum_{i \in \mathcal{S}}
    \frac{| \langle \phi_{i}, \mathbf{r}^{k} \rangle |}{\|
    \mathbf{P}^{\bot}_{\mathcal{T}^{k}} \phi_{i} \|_{2}} \nonumber\\
    &\hspace{-3mm} \overset{(a)}{=} & \hspace{-3mm} \arg   \max_{\mathcal{S} : | \mathcal{S} | =L}   \sum_{i \in \mathcal{S}} |
    \langle \phi_{i}, \mathbf{r}^{k} \rangle | \nonumber\\
     &\hspace{-3mm} = & \hspace{-3mm} \arg   \max_{\mathcal{S} : | \mathcal{S} | =L}    \|
    \mathbf{\Phi}'_{\mathcal{S}} \mathbf{y} \|_{1} = \arg   \max_{\mathcal{S} : | \mathcal{S} | =L}    \|
    \mathbf{\Phi}'_{\mathcal{S}} \mathbf{y} \|_{2}.~~~ \label{eq:13ma}
  \end{eqnarray}
  where (a) is because $k = 0$ so that $\| \mathbf{P}^{\bot}_{\mathcal{T}^{k}} \phi_{i} \|_{2} = \|
  \phi_{i} \|_{2} =1$.
  By noting that $L \leq K$,\footnote{\label{note1}Note that the average of $L$ largest elements in $\{ |\langle \phi_i', \mathbf{y} \rangle |\}_{i \in \Omega}$ must be no less than that of any other subset of $\{ |\langle \phi_i', \mathbf{y} \rangle |\}_{i \in \Omega}$ whose cardinality is no less than $L$. Hence, $\sqrt{\frac{1}{L} \sum_{i \in \mathcal{T}^1} |\langle \phi_i', \mathbf{y} \rangle |^2} \geq \sqrt{\frac{1}{K} \sum_{i \in \mathcal{T}} |\langle \phi_i', \mathbf{y} \rangle |^2}$.}
  \begin{eqnarray}
    \left\| \mathbf{\Phi}_{\mathcal{T}^{1}}' \mathbf{y} \right\|_{2} & = &  \max_{\mathcal{S} : | \mathcal{S}
    | =L} \| \mathbf{\Phi}'_{\mathcal{S}} \mathbf{y} \|_{2} \nonumber \\
    & \geq &
    \sqrt{\frac{L}{K}} \left\| \mathbf{\Phi}_{\mathcal{T}}' \mathbf{y} \right\|_{2}
    = \sqrt{\frac{L}{K}}  \left\| \mathbf{\Phi}_{\mathcal{T}}'  \mathbf{\Phi}_{\mathcal{T}}
    \mathbf{x}_{\mathcal{T}} \right\|_{2} \nonumber\\
    & \geq & \sqrt{\frac{L}{K}}  (1- \delta_{K} ) \| \mathbf{x} \|_{2}.
    \label{eq:try}
  \end{eqnarray}
  On the other hand, if no correct index is chosen at the first iteration
  (i.e., $\mathcal{T}^{1} \cap \mathcal{T}= \emptyset$), then
\begin{equation}
      \left\| \mathbf{\Phi}'_{\mathcal{T}^{1}} \mathbf{y} \right\|_{2}
    = \left\| \mathbf{\Phi}'_{\mathcal{T}^{1}}  \mathbf{\Phi}_{\mathcal{T}} \mathbf{x}_{\mathcal{T}} \right\|_{2} \overset{\text{Lemma}~\ref{lem:correlationrip} }{\leq}  \delta_{K+L} \left\| \mathbf{x} \right\|_{2}. \label{eq:36}
  \end{equation}
  This,
  however,  contradicts {\eqref{eq:try}} if $
    \delta_{K+L} < \sqrt{\frac{L}{K}}  (1- \delta_{K} )$
or
  \begin{equation}
    \label{eq:gggal} \delta_{K+L} < \frac{\sqrt{L}}{\sqrt{K} + \sqrt{L}}.
  \end{equation}
  Therefore, under (\ref{eq:gggal}), at least one correct index is chosen
  at the first iteration of MOLS.

\subsubsection{Success of the general iteration}

  Assume that MOLS has selected at least one correct index at each of the previous $k$ ($1 \leq k < K$) iterations and denote by $\ell$ the number of correct indices
  in $\mathcal{T}^{k}$. Then $
    \ell = |\mathcal{T} \cap \mathcal{T}^{k} | \geq k.$
  Also, assume that $\mathcal{T}^{k}$ does not contain all correct indices, that is, $\ell < K$. Under these assumptions, we will establish a condition that ensures MOLS to select at least one correct index at the $(k+1)$-th iteration.

  For analytical convenience, we introduce the following two quantities: i) $u_{1}$ denotes the largest value of $\frac{| \langle \phi_{i},
  \mathbf{r}^{k} \rangle |}{\| \mathbf{P}^{\bot}_{\mathcal{T}^{k}} \phi_{i} \|_{2}}$, $i
  \in \mathcal{T} \backslash \mathcal{T}^k$ and ii)
  $v_{L}$ denotes the $L$-th largest value of $\frac{| \langle
  \phi_{i}, \mathbf{r}^{k} \rangle |}{\| \mathbf{P}^{\bot}_{\mathcal{T}^{k}} \phi_{i}
  \|_{2}}$, $i \in \Omega \setminus (\mathcal{T} \cup \mathcal{T}^{k} )$. It is clear that if
\begin{equation}
  u_{1} > v_{L}, \label{eq:uv1L}
  \end{equation}
  $u_{1}$ belongs to the set of $L$ largest elements among all elements in $\big\{ \frac{| \langle
  \phi_{i}, \mathbf{r}^{k} \rangle |}{\| \mathbf{P}^{\bot}_{\mathcal{T}^{k}} \phi_{i}
  \|_{2}}\big\}_{i \in \Omega \setminus \mathcal{T}^{k}}$. Then  it follows from {\eqref{eq:golsrule11111}}
  that at least one correct
  index (i.e., the one corresponding to $u_{1}$) will be selected at the
  $(k+1)$-th iteration of MOLS. The following proposition gives a lower
  bound for $u_{1}$ and an upper bound for $v_{L}$.

  \begin{proposition}
    \label{lem:upperbound} We have
    \begin{eqnarray}
      u_{1} & \hspace{-3mm}\geq & \hspace{-3mm}\frac{1- \delta_{K + Lk -
      \ell}}{\sqrt{K - \ell}}  \left\| \mathbf{x}_{\mathcal{T} \backslash \mathcal{T}^k}  \right\|_{2},  \label{eq:small} \\
      v_{L} & \hspace{-3mm}\leq & \hspace{-3mm}\left(1 + \frac{\delta_{Lk + 1}^2}{1- \delta_{Lk} -
      \delta_{Lk+1}^{2}} \right)^{1/2}  \nonumber \\
      & & \times  \left( \delta_{L+K- \ell} +
      \frac{\delta_{L+Lk} \delta_{Lk+K- \ell}}{1- \delta_{Lk}} \right)
      \frac{\left\| \mathbf{x}_{\mathcal{T} \backslash \mathcal{T}^k}
      \right\|_{2}}{\sqrt{L}}. ~~
      \label{eq:large}
    \end{eqnarray}
  \end{proposition}

  \begin{proof}
    See Appendix~\ref{app:upperbound}.
  \end{proof}
By noting that $1 \leq k \leq \ell <K$ and $1\leq L \leq K$, we can use
  monotonicity of isometry constant to obtain
  \begin{eqnarray}
    K-\ell<LK & \Rightarrow & \delta_{K-\ell} \leq \delta_{LK}, \nonumber\\
    Lk+K-\ell\leq LK & \Rightarrow & \delta_{Lk+K-\ell} \leq \delta_{LK}, \nonumber\\
    Lk<LK & \Rightarrow & \delta_{Lk} \leq \delta_{LK}, \label{eq:monoto} \\
    Lk+1 \leq LK & \Rightarrow & \delta_{Lk+1} \leq \delta_{LK}, \nonumber\\
    L+Lk \leq LK & \Rightarrow & \delta_{L+Lk} \leq \delta_{LK}.
     \nonumber
  \end{eqnarray}
  Using {\eqref{eq:small}} and (\ref{eq:monoto}), we have
  \begin{equation}
  u_{1} \geq \frac{(1- \delta_{LK}) \| \mathbf{x}_{\mathcal{T} \backslash \mathcal{T}^k} \|_{2}}{\sqrt{K - \ell}}. \label{eq:37}
  \end{equation}
  Also, using {\eqref{eq:large}} and (\ref{eq:monoto}), we have
  \begin{eqnarray}
    \label{eq:36} v_{L}
     & \hspace{-2mm}\leq & \hspace{-2mm} \left(\hspace{-.5mm} 1 \hspace{-.5mm} + \hspace{-.5mm} \frac{\delta_{LK}^2}{1- \delta_{LK} - \delta_{LK}^{2}} \right)^{\hspace{-1mm}1/2}\hspace{-1mm} \left(\hspace{-.5mm} \delta_{LK} \hspace{-.5mm} + \hspace{-.5mm} \frac{\delta_{LK}^{2}}{1- \delta_{LK}} \right) \hspace{-1mm}  \frac{\left\| \mathbf{x}_{\mathcal{T} \backslash \mathcal{T}^k}  \right\|_{2}}{\sqrt{L}} \nonumber\\
    &\hspace{-2mm} = & \hspace{-2mm}\frac{\delta_{LK} \left\| \mathbf{x}_{\mathcal{T} \backslash \mathcal{T}^k}  \right\|_{2} }{\sqrt{L}( 1- \delta_{LK} - \delta_{LK}^{2}
    )^{1/2} (1- \delta_{LK} )^{1/2}}.
  \end{eqnarray}
From {\eqref{eq:37}} and {\eqref{eq:36}}, $u_{1} > v_{L}$ holds true whenever
  \begin{eqnarray}
    \label{eq:sufficientommp4}
    \frac{1 - \delta_{LK}}{\sqrt{K - \ell}}  > \frac{\delta_{LK}}{\sqrt{L}( 1- \delta_{LK} - \delta_{LK}^{2}
    )^{1/2} (1- \delta_{LK} )^{1/2}}.
  \end{eqnarray}
Equivalently (see Appendix~\ref{app:cond}),
  \begin{equation}
    \label{eq:k+1} \delta_{LK} < \frac{\sqrt{L}}{\sqrt{K} + 2 \sqrt{L}}.
  \end{equation}
  Therefore, under {\eqref{eq:k+1}}, MOLS selects at least one correct index
  at the $(k+1)$-th iteration.

\subsubsection{Overall condition}
So far, we have obtained condition~{\eqref{eq:gggal}} for the
success of MOLS at the first iteration and condition~{\eqref{eq:k+1}} for the success of the general iteration. We
now combine them to get an overall condition that ensures
selection of all support indices within $K$ iterations of MOLS.
Clearly the overall condition is governed the more restrictive one between {\eqref{eq:gggal}} and {\eqref{eq:k+1}}. We consider the following two cases:

  \begin{itemize}

    \item $L \geq 2$: Since $\delta_{LK} \geq \delta_{K+L}$ and also $
      \frac{\sqrt{L}}{\sqrt{K} + \sqrt{L}} > \frac{\sqrt{L}}{\sqrt{K} +
      2 \sqrt{L}},$~{\eqref{eq:k+1}} is more restrictive than {\eqref{eq:gggal}} and hence becomes the  overall condition of MOLS for this case.

   \item $L=1$: In this case, MOLS reduces to the conventional OLS algorithm and conditions \eqref{eq:gggal} and \eqref{eq:k+1} become
   \begin{equation} \label{eq:gggalx}
   \delta_{K+1} < \frac{1}{\sqrt{K} + 1}~\text{and}~
   \delta_{K} < \frac{1}{\sqrt{K} + 2},
   \end{equation}
respectively. One can easily check that both conditions in~{\eqref{eq:gggalx}} hold true if
       \begin{equation}
      \label{eq:ols1} \delta_{K+1} < \frac{1}{\sqrt{K} + 2}.
    \end{equation}
Therefore, under \eqref{eq:ols1}, OLS exactly recovers the support
of $K$-sparse signals in $K$ iterations.

  \end{itemize}

We have obtained the condition ensuring selection of all support indices within $K$ iterations of MOLS.
%
    When all support indices are selected, we have $\mathcal{T} \subseteq \mathcal{T}^{l}$
  where $l$ $(\leq K)$ denotes the number of actually performed iterations. Since $L \leq \min \{K, \frac{m}{K}\}$ by Table~\ref{tab:mols}, the number of totally selected indices of MOLS, (i.e., $lK$) does not exceed $m$, and hence the sparse signal can be recovered with a least squares (LS)
  projection:
  \begin{eqnarray}
    \mathbf{x}^{l}_{\mathcal{T}^{l}} & = & \arg   \min_{\mathbf{u}} \| \mathbf{y} - \mathbf{\Phi}_{\mathcal{T}^{l}} \mathbf{u} \|_{2} \nonumber\\
    & = & \mathbf{\Phi}^{\dag}_{\mathcal{T}^{l}} \mathbf{y}
    = \mathbf{\Phi}^{\dag}_{\mathcal{T}^{l}}  \mathbf{\Phi}_{\mathcal{T}^{l}}
    \mathbf{x}_{\mathcal{T}^{l}} = \mathbf{x}_{\mathcal{T}^{l}}. \label{eq:300}
  \end{eqnarray}
  As a result, the residual vector becomes zero ($\mathbf{r}^{l} = \mathbf{y} - \mathbf{\Phi} \mathbf{x}^{l}  = \mathbf{0}$), and hence the algorithm terminates and returns exact recovery of the sparse signal ($\hat{\mathbf{x}} = \mathbf{x}$).

\subsection{Convergence Rate} \label{sec:conv}
We can gain good insights by studying the rate of convergence of MOLS. In the following theorem, we show that the residual power of MOLS decays exponentially with the number of iterations.

\begin{theorem} \label{thm:8}
For any $0 \leq k < K$, the residual of MOLS satisfies
\begin{equation}
\|\mathbf{r}^{k + 1}\|_2^2 \leq (\alpha(k, L))^{k + 1} \|\mathbf{y}\|_2^2,
\end{equation}
where $$\alpha(k, L) := 1 - \frac{L (1- \delta_{Lk} - \delta_{Lk+1}^{2}) (1 - \delta_{K + Lk})^2}{K (1 + \delta_{L}) (1- \delta_{Lk}) (1 + \delta_{K + Lk})}.$$
\end{theorem}

The proof is given in Appendix~\ref{app:8}. Using Theorem~\ref{thm:8}, one can roughly compare the rate of convergence of OLS and MOLS. When $\mathbf{\Phi}$ has small isometry constants, the upper bound of the convergence rate of MOLS is better than that of OLS in a factor of $
 \frac{\alpha(k, 1)}{\alpha(K,L)} \big(\approx \frac{K - 1}{K - L}\big).$ We will see later in the experimental section that MOLS has a faster convergence than the OLS algorithm.

%
%
%
%
%
%

\section{Sparse Recovery with MOLS under Noise} \label{sec:V}

\subsection{Main Results}
In this section, we consider the general scenario where the measurements are contaminated with noise as
\begin{equation}
\mathbf{y} = \mathbf{\Phi x} + \mathbf{v}.  \label{eq:noizemodel}
\end{equation}
Note that in this scenario exact sparse recovery of $\mathbf{x}$ is not possible, Thus we employ the $\ell_2$-norm distortion (i.e., $\|\mathbf{x} - \hat{\mathbf{x}}\|_2$) as a performance measure and will derive conditions ensuring an upper bound for the recovery distortion.

Recall from Table~\ref{tab:mols} that MOLS runs until neither ``$\|\mathbf{r}^k\|_2 \geq \epsilon$ and $k < K$'' nor ``$Lk < K$'' is true.\footnote{The constraint $Lk \geq K$ actually ensures MOLS to select at least $K$ candidates before stopping. These candidate are then narrowed down to exact $K$ ones as the final output of the algorithm.} Since $L \geq 1$, one can see that the algorithm terminates when $\|\mathbf{r}^k\|_2 < \epsilon$ or $k = K$. In the following we will analyze the recovery distortion of MOLS based on these two termination cases. We first consider the case that MOLS is finished by the rule $\|\mathbf{r}^k\|_2 < \epsilon$. The following theorem provides an upper bound on $\|\mathbf{x} - \hat{\mathbf{x}}\|_2$ for this case.

\vspace{1mm}

\begin{theorem} \label{thm:noi1}
Consider the measurement model in~\eqref{eq:noizemodel}. If MOLS satisfies $\|\mathbf{r}^l\|_2 \leq \epsilon$ after $l$ ($< K)$ iterations and $\mathbf{\Phi}$ satisfies the RIP of orders $Ll + K$ and $2K$, then the output $\hat{\mathbf{x}}$ satisfies
\begin{equation}
\|\mathbf{x} - \hat{\mathbf{x}}\|_2 \leq \frac{2 \epsilon \sqrt{1 - \delta_{2K}} \hspace{-.5mm} + \hspace{-.5mm} 2 (\sqrt{1 - \delta_{2K}} \hspace{-.5mm} + \hspace{-.5mm} \sqrt{1 - \delta_{Ll + K}}) \|\mathbf{v}\|_2 }{\sqrt{(1 - \delta_{Ll + K})(1 + \delta_{2K})}}. \label{eq:lx}
\end{equation}
\end{theorem}
\begin{proof}
See Appendix \ref{app:noi1}.
\end{proof}

\vspace{0.1in}

Next, we consider the second case where MOLS terminates after $K$ iterations.
In this case, we parameterize the dependence on the noise $\mathbf{v}$ and the signal $\mathbf{x}$ with two quantities: i) the signal-to-noise ratio (SNR) and ii) the minimum-to-average ratio (MAR)~\cite{fletcher2012orthogonal} which are defined as
\begin{equation}
{snr}:= \frac{\|\mathbf{\Phi x}\|_2^2}{\|\mathbf{v}\|_2^2}~\text{and}~\kappa: = \frac{\min_{j \in \mathcal{T}} |x_{j}|}{\|\mathbf{x}\|_2 /{\sqrt K}}, \label{eq:snrmar}
\end{equation}
respectively.
 The following theorem provides an upper bound on the $\ell_2$-norm of the recovery distortion of MOLS.

 \vspace{1mm}

\begin{theorem} \label{thm:noi} Consider the measurement model in~\eqref{eq:noizemodel}. If the measurement matrix $\mathbf{\Phi}$ satisfies \eqref{eq:jjjjffffa} and the SNR satisfies
  \begin{equation}
    \left\{\begin{array}{ll}
      \sqrt{snr} \geq \frac{ 2 (1 + \delta_{K + 1})}{\kappa(1 - (\sqrt{K} + 2) \delta_{K + 1}) } \sqrt K, & L=1,\\
      \sqrt{snr} \geq \frac{(\sqrt L + 1) (1 + \delta_{LK}) }{\kappa  (\sqrt L - (\sqrt{K} + 2 \sqrt L) \delta_{LK})} \sqrt K, & L>1,
    \end{array}\right. \label{eq:jjjjffffaa}
  \end{equation}
then MOLS chooses all support indices in $K$ iterations (i.e., ${\mathcal{T}}^K \supseteq \mathcal{T}$) and generates an estimate of $\mathbf{x}$ satisfying
  \begin{equation}
    \hspace{.5mm}\left\{\begin{array}{ll}
      \|\hat{\mathbf{x}} \hspace{-.25mm} - \hspace{-.25mm} \mathbf{x}\|_2 \leq \frac{\|\mathbf{v}\|_2}{\sqrt{1 - \delta_{K}}}, & L=1,\\
      \|\hat{\mathbf{x}} \hspace{-.25mm} - \hspace{-.25mm} \mathbf{x}\|_2 \leq \hspace{-.5mm} \big(1 \hspace{-.5mm} + \hspace{-.5mm} \sqrt{\frac{1 - \delta_{2K}}{1 - \delta_{LK}}} \big) \frac{2 \|\mathbf{v}\|_2}{\sqrt{1 + \delta_{2K}}}, & L>1.
    \end{array}\right. \label{eq:stable}
\end{equation}
\end{theorem}

One can interpret from Theorem~\ref{thm:noi} that MOLS can catch all support indices of $\mathbf{x}$ in $K$ iterations when the SNR scales linearly with the sparsity $K$. In particular, for the special case of $L = 1$, the algorithm exactly recovers the support of $\mathbf{x}$ (i.e., $\hat{\mathcal{T}} = \mathcal{T}$). It is worth noting that the SNR being proportional to $K$ is necessary for exact support recovery with OLS. In fact, there exist a measurement matrix $\mathbf{\Phi}$ satisfying \eqref{eq:jjjjffffa} and a $K$-sparse signal, for which the OLS algorithm fails to recover the support of the signal under
\begin{equation}
snr \geq K.
\end{equation}
\vspace{-4mm}
\begin{example}
Consider an identity matrix $\mathbf{\Phi}^{m \times m}$, a $K$-sparse signal $\mathbf{x} \in \mathcal{R}^m$ with all nonzero elements equal to one, and an $1$-sparse noise vector $\mathbf{v} \in \mathcal{R}^m$ as follows,
    \begin{equation}
  \mathbf{\Phi} = \left[ \begin{array}{cccc}
    1 &      &        &   \\
      & 1    &        &   \\
      &      & \ddots &   \\
      &      &        & 1
  \end{array} \right]\hspace{-1mm},~
  \mathbf{x} = \left[ \begin{array}{c}
    1\\
    \vdots\\
    1\\
    0\\
    \vdots \\
    0
  \end{array} \right]\hspace{-1mm}, ~\text{and}~
  \mathbf{v} = \left[ \begin{array}{c}
    0\\
    \vdots \\
    0 \\
    1
  \end{array} \right]\hspace{-1mm}. \nonumber
\nonumber
\end{equation}
Then the measurements are given by
\begin{equation}
\mathbf{y} = ~~\overbrace {\hspace{-2mm}\left[ {\begin{array}{*{20}c}
\hspace{-1mm}1 & \cdots & 1 \hspace{-2mm} \\
\end{array}} \right.}^{K} ~
\overbrace {\hspace{-2mm}\left. {\begin{array}{*{20}c}
0 & \cdots & 0 \hspace{-2mm} \\
\end{array}} \right.}^{m - K - 1} \left. {\begin{array}{*{20}c}
1\\
\end{array}} \hspace{-1.25mm}\right]'\hspace{-1mm}.\nonumber
\end{equation}
In this case, we have $\delta_{K + 1} = 0$ (so that condition \eqref{eq:jjjjffffa} is fulfilled) and $snr = K$; however, OLS may fail to recover the support of $\mathbf{x}$. Specifically, OLS is not guaranteed to make a correct selection at the first iteration.
\end{example}
%

%
%
%

\subsection{Proof of Theorem~\ref{thm:noi}} \label{sec:proofoftheorem4}

Our proof of Theorem~\ref{thm:noi} extends the proof technique in
\cite[Theorem 3.4 and 3.5]{wang2012Generalized} (which studied the
recovery condition for the gOMP algorithm in the noiseless
situation) by considering the measurement noise. We mention
that~\cite[Theorem 4.2]{wang2012Generalized} also provided a noisy case analysis
based on the $\ell_2$-norm distortion of signal recovery, but the
corresponding result is far inferior to the result established in
Theorem~\ref{thm:noi}. Indeed, while the result
in~\cite{wang2012Generalized} suggested a recovery distortion upper
bounded by $\mathcal{O}(\sqrt K) \|\mathbf{v}\|_2$, our result shows
that the recovery distortion with MOLS is at most proportional to
the noise power. The result in Theorem~\ref{thm:noi} is also closely
related to the results in~\cite{wu2013exact,chang2014improved}, in
which the researchers considered the OMP algorithm with data driven
stopping rules (i.e., residual based stopping rules), and
established conditions for exact support recovery that depend on the
minimum magnitude of nonzero elements of input signals. It can be
shown that the results of~\cite{wu2013exact,chang2014improved}
essentially require a same scaling law of the SNR as the result in
Theorem~\ref{thm:noi}.

The key idea in the proof is to derive a condition ensuring MOLS to select at least one good index at each iterations. As long as at least one good index is chosen in each iteration, all support indices will be included in $K$ iterations of MOLS (i.e., $\mathcal{T}^K \supseteq \mathcal{T}$) and consequently the algorithm produces a stable recovery of $\mathbf{x}$.

\begin{proposition} \label{prop:r2}
Consider the measurement model in~\eqref{eq:noizemodel}. If the measurement matrix $\mathbf{\Phi}$ satisfies the RIP of order $LK$, then MOLS satisfies \eqref{eq:stable} provided that $\mathcal{T}^K \supseteq \mathcal{T}$.
\end{proposition}
\begin{proof}
See Appendix \ref{app:r2}.
\end{proof}

\vspace{0.1in}

Now we proceed to derive the condition ensuring the success of MOLS at each iteration.
Again, the notion ``success'' means that MOLS selects at least one good index at this iteration. We first derive a condition for the success of MOLS at the first iteration. Then we assume that MOLS has been successful in the previous $k$ iterations and derive a condition guaranteeing MOLS to make a success as well at the $(k + 1)$-th iteration. Finally, we combine these two conditions to obtain an overall condition for MOLS.

\subsubsection{Success at the first iteration}
  From \eqref{eq:13ma}, we know that at the first iteration,  MOLS selects the set $
    \mathcal{T}^{1}$ of $L$ indices such that $
    \| \mathbf{\Phi}'_{\mathcal{T}^{1}} \mathbf{y} \|_{2} = \max_{\mathcal{S} : | \mathcal{S}
    | =L} \| \mathbf{\Phi}'_{\mathcal{S}} \mathbf{y} \|_{2}.$
Since $L \leq K$,
  \begin{eqnarray}
    \left\| \mathbf{\Phi}_{\mathcal{T}^{1}}' \mathbf{y} \right\|_{2}
    \hspace{-2mm}& \geq &\hspace{-2mm}
    \sqrt{\frac{L}{K}} \left\| \mathbf{\Phi}_{\mathcal{T}}' \mathbf{y} \right\|_{2}
    =
    \sqrt{\frac{L}{K}} \left\| \mathbf{\Phi}_{\mathcal{T}}' \mathbf{\Phi x} + \mathbf{\Phi}_{\mathcal{T}}' \mathbf{v} \right\|_{2}
    \nonumber\\
    \hspace{-2mm}& \overset{(a)}{\geq} &\hspace{-2mm} \sqrt{\frac{L}{K}}  \left(\| \mathbf{\Phi}_{\mathcal{T}}'  \mathbf{\Phi}_{\mathcal{T}}
    \mathbf{x}_{\mathcal{T}} \|_{2} - \|\mathbf{\Phi}_{\mathcal{T}}' \mathbf{v}\|_{2} \right) \nonumber\\
    \hspace{-2mm}& \overset{(b)}{\geq} &\hspace{-2mm} \sqrt{\frac{L}{K}}  \left((1- \delta_{K} ) \| \mathbf{x} \|_{2} -  \sqrt{1 + \delta_K} \|\mathbf{v}\|_{2} \right),~~~~~
    \label{eq:try5}
  \end{eqnarray}
  where (a) is from the triangle inequality and (b) is from the RIP and Lemma~\ref{lem:rip5}.

  On the other hand, if no correct index is chosen at the first iteration
  (i.e., $\mathcal{T}^{1} \cap \mathcal{T}= \emptyset$), then
  \begin{eqnarray}
    \left\| \mathbf{\Phi}'_{\mathcal{T}^{1}} \mathbf{y} \right\|_{2}
    &=& \left\| \mathbf{\Phi}'_{\mathcal{T}^{1}}  \mathbf{\Phi}_{\mathcal{T}} \mathbf{x}_{\mathcal{T}}  + \mathbf{\Phi}_{\mathcal{T}}' \mathbf{v} \right\|_{2}
    \nonumber \\
    &\leq& \left\| \mathbf{\Phi}'_{\mathcal{T}^{1}}  \mathbf{\Phi}_{\mathcal{T}} \mathbf{x}_{\mathcal{T}}\|_2  + \| \mathbf{\Phi}_{\mathcal{T}}' \mathbf{v} \right\|_{2}
    \nonumber \\
    & \overset{\text{Lemma}~\ref{lem:correlationrip},~\ref{lem:rip5} }{\leq} & \delta_{K+L} \left\| \mathbf{x} \right\|_{2} + \sqrt{1 + \delta_K} \|\mathbf{v}\|_{2}.~~~ \label{eq:365}
  \end{eqnarray}
  This,
  however,  contradicts {\eqref{eq:try5}} if
\begin{equation}
   \hspace{-.5mm} \delta_{K+L} \hspace{-.5mm} \| \mathbf{x} \|_{2} \hspace{-.25mm} + \hspace{-.25mm} \sqrt{\hspace{-.5mm}1 \hspace{-1mm} + \hspace{-.75mm} \delta_K} \|\mathbf{v}\|_{2} \hspace{-.5mm} < \hspace{-.75mm} \sqrt{\hspace{-.5mm}\frac{L}{K}} \hspace{-.5mm} \left( \hspace{-.5mm}(1 \hspace{-.5mm}-\hspace{-.5mm} \delta_{K} ) \| \mathbf{x} \|_{2}\hspace{-.75mm} - \hspace{-.75mm} \sqrt{\hspace{-.5mm}1 \hspace{-.75mm} + \hspace{-.5mm} \delta_K} \|\mathbf{v}\|_{2} \hspace{-.5mm} \right)\hspace{-0.5mm}.\nonumber
  \end{equation}
  Equivalently,
  \begin{eqnarray}
  \hspace{-1.25mm}\left( \hspace{-.75mm} (1  \hspace{-.5mm} -  \hspace{-.5mm} \delta_{K + 1})  \hspace{-.25mm} \sqrt{\hspace{-.5mm}\frac{L}{K}}  \hspace{-.5mm} -  \hspace{-.5mm} \delta_{K + L}  \hspace{-.5mm} \right) \hspace{-1mm}  \frac{ \|\mathbf{x}\|_2}{ \|\mathbf{v}\|_2}  \hspace{-.75mm} >  \hspace{-.75mm} \left( \hspace{-.5mm} 1  \hspace{-.5mm} +  \hspace{-.5mm} \sqrt{\hspace{-.5mm}\frac{L}{K}}  \right)  \hspace{-1mm} \sqrt{1  \hspace{-.5mm} +  \hspace{-.5mm} \delta_{K + L}}.\hspace{-1mm}  \label{eq:11000121}
  \end{eqnarray}
Furthermore, since
\begin{equation}
 \|\mathbf{\Phi x}\|_2 \overset{\text{RIP}}{\leq} \sqrt{1 + \delta_{K}} \|\mathbf{x}\|_2 \overset{\text{Lemma}~\ref{lem:mono}}{\leq} \sqrt{1 + \delta_{K + L}} \|\mathbf{x}\|_2,
 \end{equation} using \eqref{eq:snrmar} we can show that \eqref{eq:11000121} holds true if
\begin{equation}
\sqrt{snr} > \frac{(1 + \delta_{K + L}) (\sqrt K + \sqrt L)}{\sqrt L - (\sqrt K + \sqrt L) \delta_{K + L}}.
    \label{eq:k+1n}
  \end{equation}
  Therefore, under (\ref{eq:k+1n}), at least one correct index is chosen
  at the first iteration of MOLS.

\subsubsection{Success at the $(k+1)$-th iteration}

  Similar to the analysis of MOLS in the noiseless case in Section~\ref{sec:IV}, we assume that MOLS selects at least one correct index at each of the previous
  $k$ ($1 \leq k < K$) iterations and denote by $\ell'$ the number of correct indices
  in $\mathcal{T}^{k}$. Then, $
    \ell' = |\mathcal{T} \cap \mathcal{T}^{k} | \geq k.$
  Also, we assume that $\mathcal{T}^{k}$ does not contain all correct indices ($\ell' <K$).
  Under these assumptions, we derive a condition that ensures MOLS
 to select at least one correct index at the $(k+1)$-th iteration.

We introduce two quantities that are useful for stating results. Let $u'_{1}$ denote the largest value of $\frac{| \langle \phi_{i},
  \mathbf{r}^{k} \rangle |}{\| \mathbf{P}^{\bot}_{\mathcal{T}^{k}} \phi_{i} \|_{2}}$, $i
  \in \mathcal{T}$ and let
  $v'_{L}$ denote the $L$-th largest value of $\frac{| \langle
  \phi_{i}, \mathbf{r}^{k} \rangle |}{\| \mathbf{P}^{\bot}_{\mathcal{T}^{k}} \phi_{i}
  \|_{2}}$, $i \in \Omega \setminus (\mathcal{T} \cup \mathcal{T}^{k} )$. It is clear that if \begin{equation}
  u'_{1} > v'_{L},
  \end{equation}
  then $u'_{1}$ belongs to the set of $L$ largest elements among all elements in $\big\{ \frac{| \langle
  \phi_{i}, \mathbf{r}^{k} \rangle |}{\| \mathbf{P}^{\bot}_{\mathcal{T}^{k}} \phi_{i}
  \|_{2}}\big\}_{i \in \Omega \setminus \mathcal{T}^{k}}$. Then it follows from {\eqref{eq:golsrule11111}}
  that at least one correct
  index (i.e., the one corresponding to $u'_{1}$) will be selected at the
  $(k+1)$-th iteration. The following proposition gives a lower
  bound for $u'_{1}$ and an upper bound for $v'_{L}$.

 \vspace{1mm}
  \begin{proposition}
    \label{prop:upperbound15} We have
    \begin{eqnarray}
u'_{1} \hspace{-3mm} &\geq& \hspace{-3mm} \frac{( 1 - \delta_{K + Lk -
      \ell'} )  \left\| \mathbf{x}_{\mathcal{T} \backslash \mathcal{T}^k}  \right\|_{2}  \hspace{-1mm} - \hspace{-1mm} \sqrt{1 + \delta_{K + Lk - \ell'}} \left\| \mathbf{v} \right\|_{2}}{\sqrt{K - \ell'}},~~~~~  \label{eq:small5} \\
      v'_{L}
      \hspace{-3mm} & \leq & \hspace{-3mm} \frac{1}{\sqrt{L}} \left( \left( \delta_{L+K- \ell'} +
      \frac{\delta_{L+Lk} \delta_{Lk+K- \ell'}}{1- \delta_{Lk}} \right)
      {\left\| \mathbf{x}_{\mathcal{T} \backslash \mathcal{T}^k}
      \right\|_{2}}  \right. \nonumber \\
      && \hspace{-2mm} \left. + \sqrt{1 + \delta_{L + Lk}} \|\mathbf{v}\|_2 \right) \left(1 + \frac{\delta_{Lk + 1}^2}{1- \delta_{Lk} - \delta_{Lk+1}^{2}} \right)^{\hspace{-1mm}1/2} \hspace{-1mm}.
      \label{eq:large5}
    \end{eqnarray}
  \end{proposition}

  \begin{proof}
    See Appendix~\ref{app:upperbound1}.
  \end{proof}

By noting that $1 \leq k \leq \ell' <K$ and $1\leq L \leq K$, and also using Lemma~\ref{lem:mono}, we
  have
  \begin{eqnarray}\label{eq:monoto55}
    K-\ell' <LK & \Rightarrow & \delta_{K-\ell'} \leq \delta_{LK}, \nonumber\\
    Lk+K-\ell' \leq LK & \Rightarrow & \delta_{Lk+K-\ell'} \leq \delta_{LK}.
  \end{eqnarray}
  From (\ref{eq:monoto}), {\eqref{eq:large5}}, and (\ref{eq:monoto55}), we have
  \begin{eqnarray}
  \hspace{-3mm} v_{L}
    \hspace{-3mm} & \leq & \hspace{-3mm} \frac{1}{\sqrt{L}} \left(1 +  \frac{\delta_{LK}^2}{1- \delta_{LK}
    - \delta_{LK}^{2}} \right)^{1/2} \nonumber \\
   \hspace{-8mm}  & & \hspace{-2mm} \times \left( \left( \delta_{LK} +
    \frac{\delta_{LK}^{2}}{1- \delta_{LK}} \right) {\left\| \mathbf{x}_{\mathcal{T} \backslash \mathcal{T}^k} \right\|_{2}} + \sqrt{1 + \delta_{LK}} \|\mathbf{v}\|_2 \right) \nonumber \\
    \hspace{-8mm} & = & \hspace{-3.5mm} \frac{1}{\sqrt{L}} \hspace{-.75mm} \left(\hspace{-.75mm} \frac{1 \hspace{-.5mm} - \hspace{-.5mm} \delta_{LK}}{1 \hspace{-.5mm} - \hspace{-.75mm} \delta_{LK} \hspace{-.75mm}
    - \hspace{-.75mm} \delta_{LK}^{2}} \hspace{-.75mm} \right) ^{\hspace{-1.25mm} 1/2} \hspace{-1.5mm}  \left( \hspace{-.75mm}
    \frac{\delta_{LK} \hspace{-.5mm} {\left\| \mathbf{x}_{\mathcal{T} \backslash \mathcal{T}^k} \right\|_{2}}} {1- \delta_{LK}} \hspace{-.75mm} + \hspace{-1mm} \sqrt{\hspace{-.25mm} 1 \hspace{-.75mm} + \hspace{-.5mm} \delta_{LK}} \|\mathbf{v}\|_2 \hspace{-1mm} \right)\hspace{-.75mm}.\label{eq:3655} \nonumber \\
  \end{eqnarray}
  Also, from (\ref{eq:monoto}), {\eqref{eq:small5}}, and (\ref{eq:monoto55}), we have
\begin{eqnarray}\label{eq:3755}
  u_{1} \hspace{-.5mm} \geq \hspace{-.5mm} \frac{1}{\sqrt{K - \ell'}} \left((1 \hspace{-.5mm} - \hspace{-.5mm} \delta_{LK}) \|
    \mathbf{x}_{\mathcal{T} \backslash \mathcal{T}^k} \|_{2}\hspace{-.5mm} - \hspace{-.5mm} \sqrt{1 \hspace{-.5mm} + \hspace{-.5mm} \delta_{LK}} \|\mathbf{v}\|_2\right)\hspace{-.5mm}.
  \end{eqnarray}
  From {\eqref{eq:3655}} and {\eqref{eq:3755}}, $u'_{1} > v'_{L}$ can be guaranteed by
  \begin{eqnarray}
    \label{eq:sufficientommp45}
  \hspace{-7.5mm}  &&\frac{1}{\sqrt{K - \ell'}} \left((1- \delta_{LK}) \|
    \mathbf{x}_{\mathcal{T} \backslash \mathcal{T}^k} \|_{2} - \sqrt{1 + \delta_{LK}} \|\mathbf{v}\|_2\right)  \nonumber \\
  \hspace{-8.5mm}  && > \hspace{-.5mm}  \frac{1}{\sqrt{L}} \hspace{-.75mm} \left(\hspace{-.75mm} \frac{1 \hspace{-.5mm} - \hspace{-.5mm} \delta_{LK}}{1 \hspace{-.5mm} - \hspace{-1mm} \delta_{LK}
    - \delta_{LK}^{2}} \hspace{-.75mm} \right) ^{\hspace{-1.25mm} 1/2} \hspace{-1.5mm}  \left( \hspace{-.75mm}
    \frac{\delta_{LK} {\left\| \mathbf{x}_{\mathcal{T} \backslash \mathcal{T}^k} \right\|_{2}}} {1- \delta_{LK}} \hspace{-.75mm} + \hspace{-.75mm} \sqrt{1 \hspace{-.5mm} + \hspace{-.5mm} \delta_{LK}} \|\mathbf{v}\|_2 \hspace{-1mm} \right)\hspace{-1mm}, \nonumber \\ \hspace{-7.5mm} \label{eq:48o}
  \end{eqnarray}
  which is true under (see Appendix~\ref{app:cond5})
\begin{equation}
    \label{eq:k+15}
\sqrt{snr} \geq \frac{(\sqrt L + 1) (1 + \delta_{LK}) }{\kappa  (\sqrt L - (\sqrt{K} + 2 \sqrt L) \delta_{LK})} \sqrt K,
 \end{equation}
  Therefore, under {\eqref{eq:k+15}}, MOLS selects at least one correct index
  at the $(k+1)$-th iteration.

{\it 3) Overall condition}: Thus far we have obtained condition~{\eqref{eq:k+1n}} for the
success of MOLS at the first iteration and condition
{\eqref{eq:k+15}} for the success of the general iteration. We
now combine them to get an overall condition of MOLS ensuring selection of all support indices in $K$ iterations.
Clearly the overall condition can be the more restrictive one between {\eqref{eq:k+1n}} and {\eqref{eq:k+15}}. We consider the following two cases.
  \begin{itemize}
    \item $L \geq 2$: Since $\delta_{LK} \geq \delta_{K+L}$, {\eqref{eq:k+15}} is more restrictive than~{\eqref{eq:k+1n}} and becomes the overall condition.

   \item $L=1$: In this case, the MOLS algorithm reduces to the conventional OLS algorithm. Since $\delta_{K + 1} \geq \delta_K$, one can verify that both {\eqref{eq:k+1n}} and {\eqref{eq:k+15}} hold true under
   \begin{equation}
     \sqrt{snr} \geq \frac{ 2 (1 + \delta_{K + 1})}{\kappa(1 - (\sqrt{K} + 2) \delta_{K + 1}) } \sqrt K. \label{eq:yuuu}
    \end{equation}
    Therefore, \eqref{eq:yuuu} ensures selection of all support indices in $K$ iterations of OLS.
    \end{itemize}
    The proof is now complete.

\section{Empirical Results} \label{sec:VI}

In this section, we empirically study the performance of MOLS
in recovering sparse signals. We consider both the noiseless and noisy scenarios. In the noiseless case, we adopt the testing
strategy in~\cite{candes2005error,dai2009subspace,wang2012Generalized} which
measures the performance of recovery algorithms by testing their
empirical frequency of exact reconstruction of sparse signals, while in the noisy case, we employ the mean square error (MSE) as a metric to evaluate the recovery performance. For comparative
purposes, we consider the following recovery approaches in our
simulation:
\begin{enumerate}
  \item OLS and MOLS;

  \item OMP;

  \item StOMP (\url{http://sparselab.stanford.edu/});

  \item ROMP  (\url{http://www.cmc.edu/pages/faculty/DNeedell});

  \item CoSaMP (\url{http://www.cmc.edu/pages/faculty/DNeedell});

  \item BP (or BPDN for the noisy case) (\url{http://
cvxr.com/cvx/});

  \item Iterative reweighted LS (IRLS);
  \item Linear minimum MSE (LMMSE) estimator.
\end{enumerate}
In
each trial, we construct an $m \times n$ matrix (where $m = 128$ and
$n = 256$) with entries drawn independently from a Gaussian
distribution with zero mean and $\frac{1}{m}$ variance. For each
value of $K$ in $\{5, \cdots, 64\}$, we generate a $K$-sparse signal
of size $n \times 1$ whose support is chosen uniformly at random and
nonzero elements are 1) drawn independently from a standard Gaussian
distribution, or 2) chosen randomly from the set $\{\pm 1\}$.
We refer to the two types of signals as the sparse Gaussian signal
and the sparse $2$-ary pulse amplitude modulation ($2$-PAM) signal,
respectively. We mention that reconstructing sparse $2$-PAM signals is a
particularly challenging case for OMP and OLS.

\begin{figure}[t]
\centering
\subfigure[Sparse Gaussian signal.]
{\includegraphics[width = 88mm] {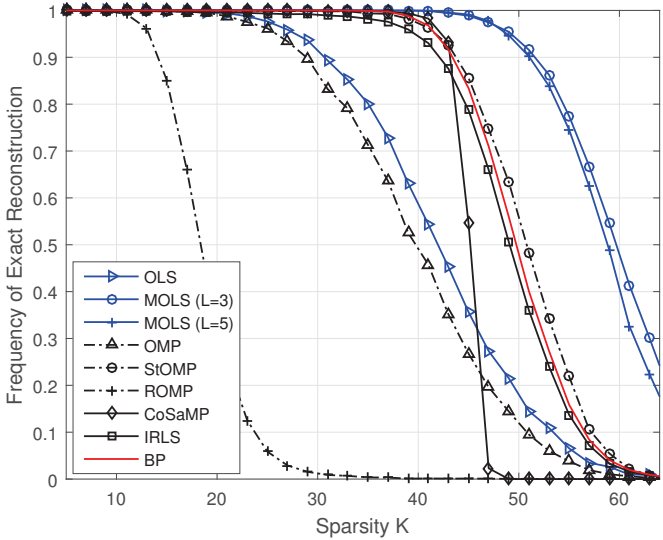}
\label{fig:Gauss}}
\subfigure[Sparse $2$-PAM signal]
{\includegraphics[width = 88mm] {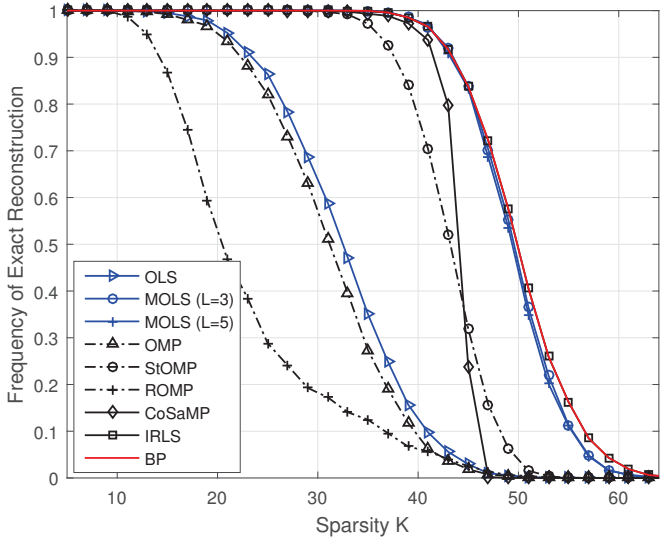} \label{fig:01}}
\caption{Frequency of exact recovery of sparse signals as a function of $K$.}
\label{fig:ER}
\end{figure}

\begin{figure}[h!]
\centering
\subfigure[Running time (Gaussian signals).]
{\includegraphics[scale =.465] {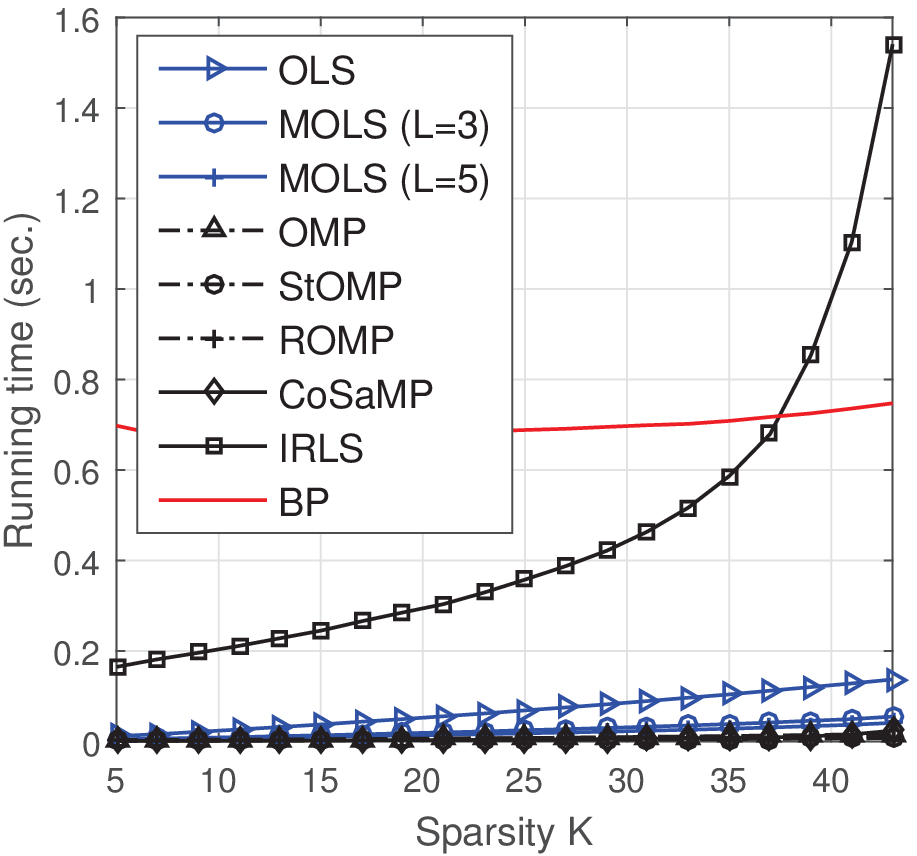} \label{fig:subfig1}}
\hspace{-2mm}\subfigure[Running time ($2$-PAM signals).]
{\includegraphics[scale =.465] {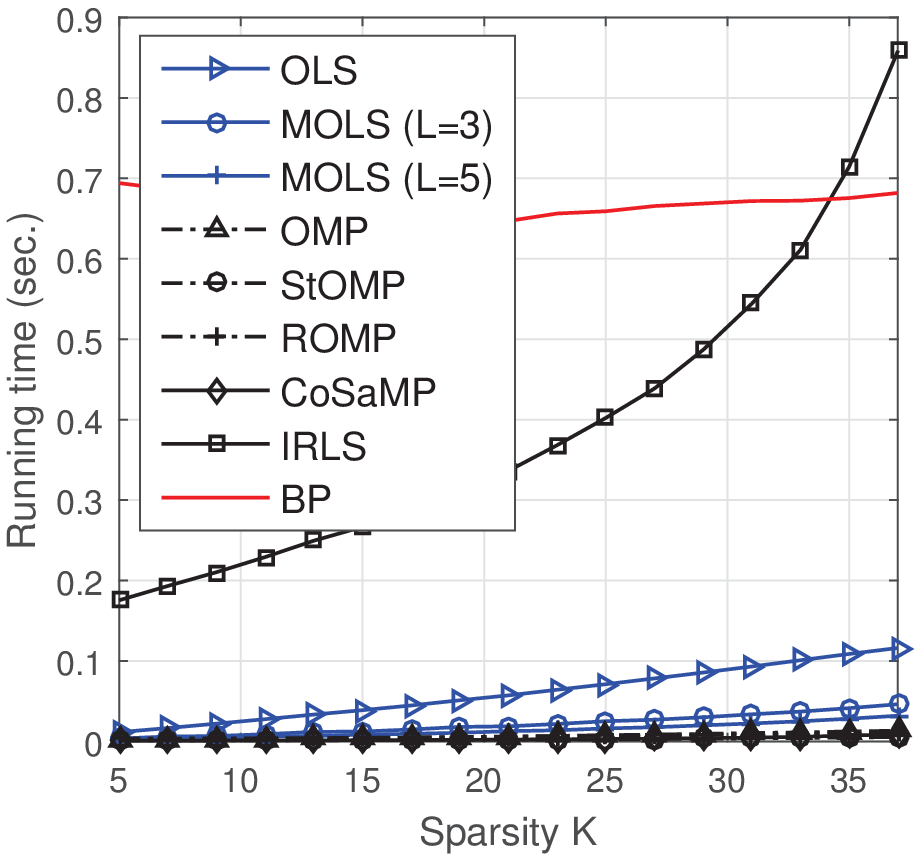} \label{fig:subfig2}}

\subfigure[Running time without BP and IRLS (Gaussian signals).]
{\includegraphics[scale =.45] {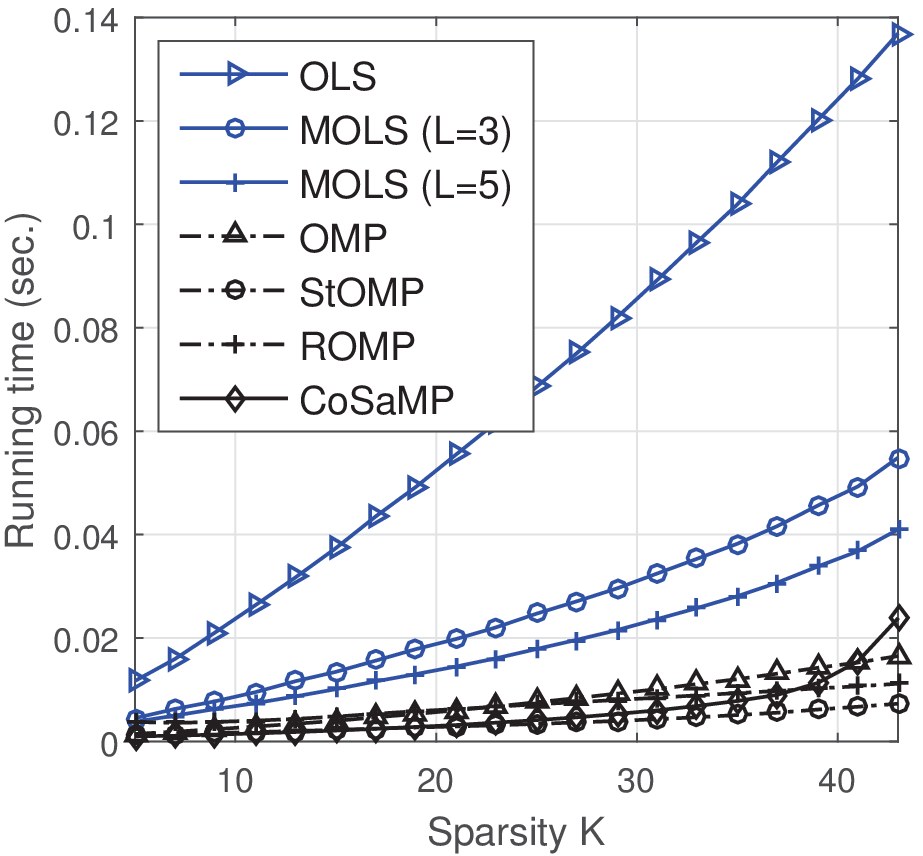} \label{fig:subfig3}}
\hspace{-0mm}\subfigure[Running time without BP and IRLS ($2$-PAM signals).]
{\includegraphics[scale =.45] {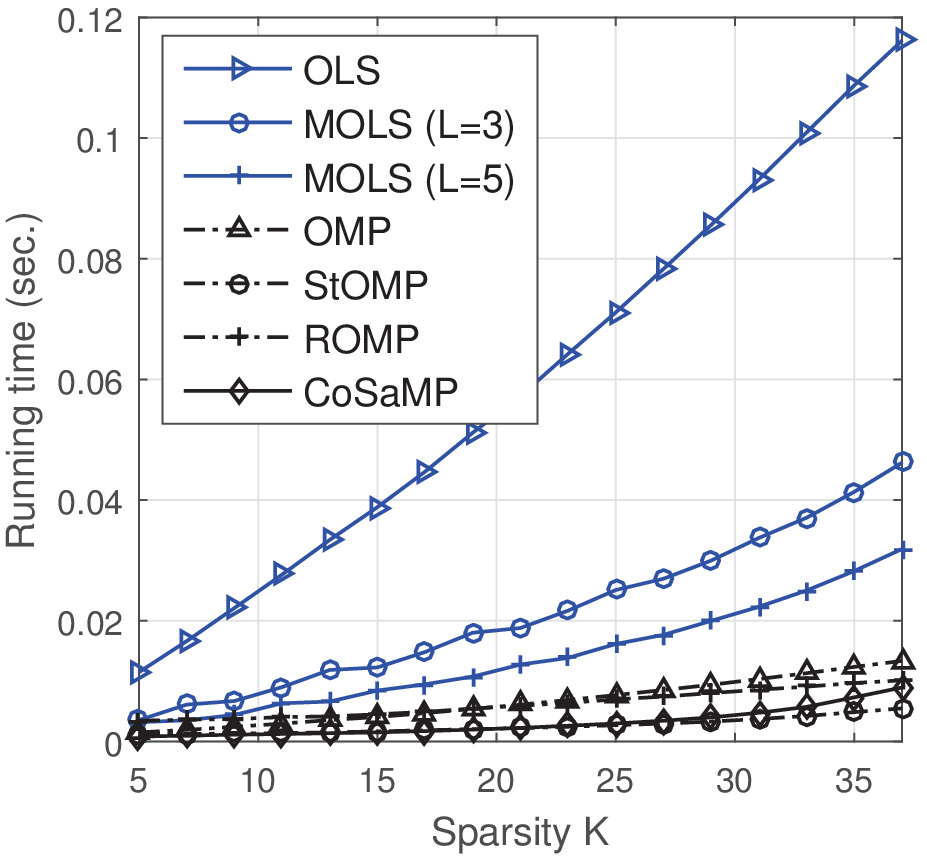} \label{fig:subfig4}}

\subfigure[\# of iterations (Gaussian signals).]
{\includegraphics[scale =.465] {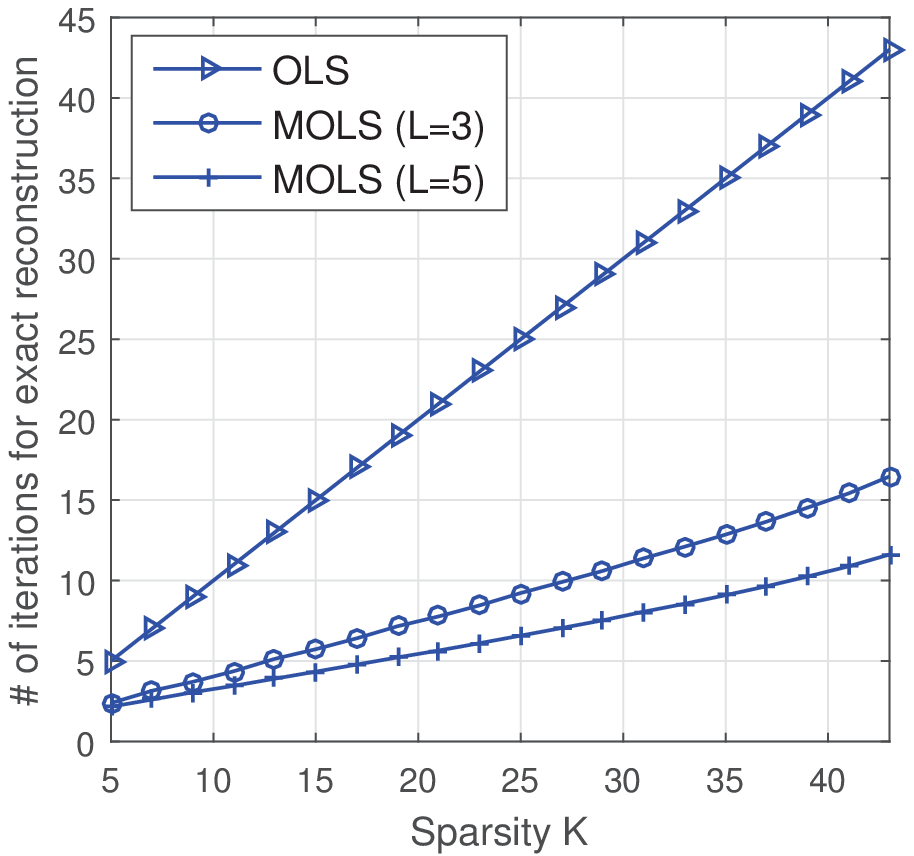} \label{fig:subfig5}}
\hspace{-1.5mm}\subfigure[\# of iterations ($2$-PAM signals).]
{\includegraphics[scale =.465] {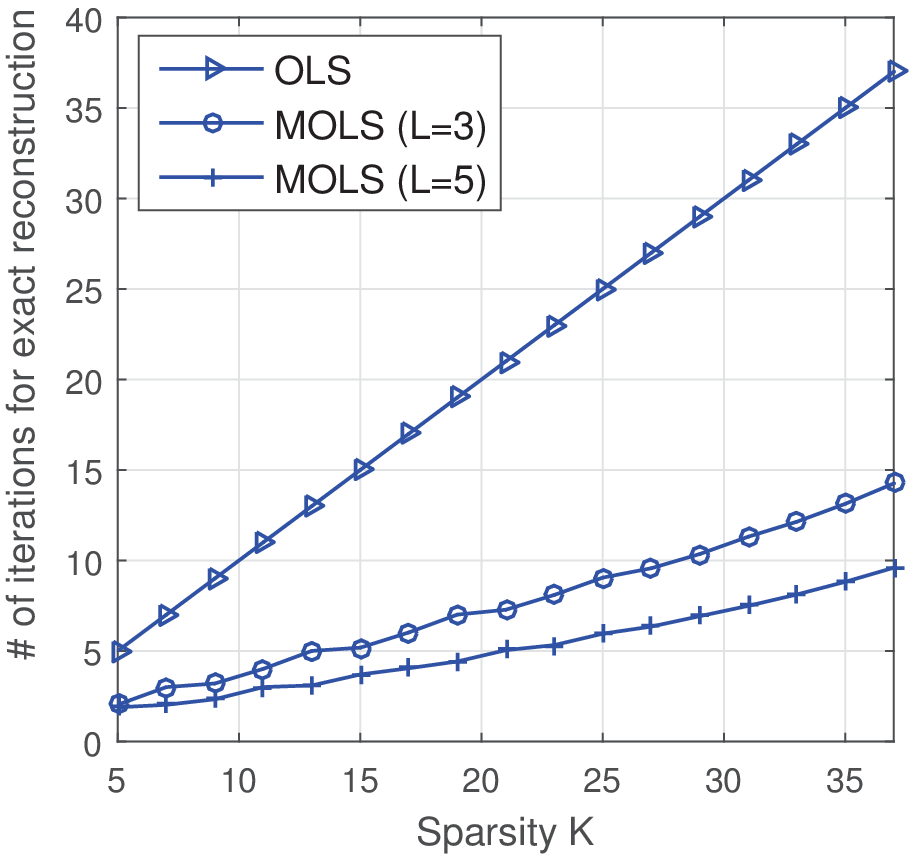} \label{fig:subfig6}}
\caption{Running time and number
of iterations for exact reconstruction of $K$-sparse Gaussian and
$2$-PAM signals.} \label{fig:iter}
\end{figure}


In the noiseless case, we perform $2,000$ independent trials for each recovery approach and plot the empirical frequency of exact reconstruction as a function of the sparsity level. By comparing the maximal
sparsity level, i.e., the so called {\it critical sparsity}~\cite{dai2009subspace}, of sparse signals at which exact
reconstruction is always ensured, recovery accuracy of different
algorithms can be compared empirically.\footnote{\label{foot:1}Note that for MOLS, the selection
parameter should obey $L \leq K$. We thus choose $L = 3, 5$ in our
simulation. Interested reader may try other options. We suggest to choose $L$ to be small integers and have empirically confirmed that choices of $L = 2, 3, 4, 5$ generally lead to similar recovery performance. For StOMP, there are two thresholding strategies: false
alarm control (FAC) and  false discovery control (FDC)~\cite{donoho2006sparse}. We
exclusively use FAC, since the FAC outperforms FDC. For OMP and OLS, we run the algorithm for exact $K$ iterations before stopping. For CoSaMP: we set the maximal iteration number to $50$ to avoid repeated iterations. We implement IRLS (with $p = 1$) as featured in~\cite{chartrand2008iteratively}.} As shown in Fig.~\ref{fig:ER}, for both sparse Gaussian and sparse $2$-PAM signals,
the MOLS algorithm outperforms other greedy approaches with respect
to the critical sparsity. Even when compared to the BP and IRLS methods, the MOLS algorithm still exhibits very competitive reconstruction
performance. For the Gaussian case, the critical sparsity of MOLS is $43$, which is higher than that of BP and IRLS, while for the $2$-PAM case, MOLS, BP and IRLS have almost identical critical sparsity (around~$37$).

In Fig.~\ref{fig:iter}, we plot the running time and the number of
iterations for exact reconstruction of
$K$-sparse Gaussian and $2$-PAM signals as a function of $K$. The
running time is measured using the MATLAB program under the
$28$-core $64$-bit processor, $256$Gb RAM, and Windows Server $2012$ R$2$
environments. Overall, we observe that for both sparse Gaussian and $2$-PAM cases, the running time of BP and IRLS is longer than that of OMP, CoSaMP, StOMP and MOLS. In particular, the running time of BP is more than one order of magnitude higher than the rest of algorithms require. This is because the complexity of BP is a quadratic function of the number of measurements ($\mathcal{O}(m^2 n^{3/2})$)~\cite{nesterov1994interior}, while that of greedy algorithms is $\mathcal{O}(Kmn)$.
Moreover, the running time of MOLS is roughly two to three times as much as that of OMP.
We also observe that the number of iterations of MOLS for exact reconstruction is much smaller than that of the OLS algorithm since MOLS can include more than one support index at a time.
The associated running time of MOLS is also much less than that of OLS.

In Fig. \ref{fig:varym}, by varying the number of measurements $m$,
we plot empirical frequency of exact reconstruction of $K$-sparse
Gaussian signals as a function of $m$. We consider the sparsity
level $K = 45$, for which none of the reconstruction methods in
Fig.~\ref{fig:Gauss} are guaranteed to perform exact recovery.
Overall, we observe that the performance comparison among all
reconstruction methods is similar to Fig.~\ref{fig:Gauss} in that
MOLS performs the best and OLS, OMP and ROMP perform worse than
other methods. Moreover, for all reconstruction methods under test,
the frequency of exact reconstruction improves as the number of
measurements increases. In particular, MOLS roughly requires $m \geq
135$ to ensure exact recovery of sparse signals, while BP, CoSaMP and StOMP seem to always succeed when $m \geq 150$.

\begin{figure}[t]
\centering
\hspace{-1mm}{\includegraphics[width = 90 mm]{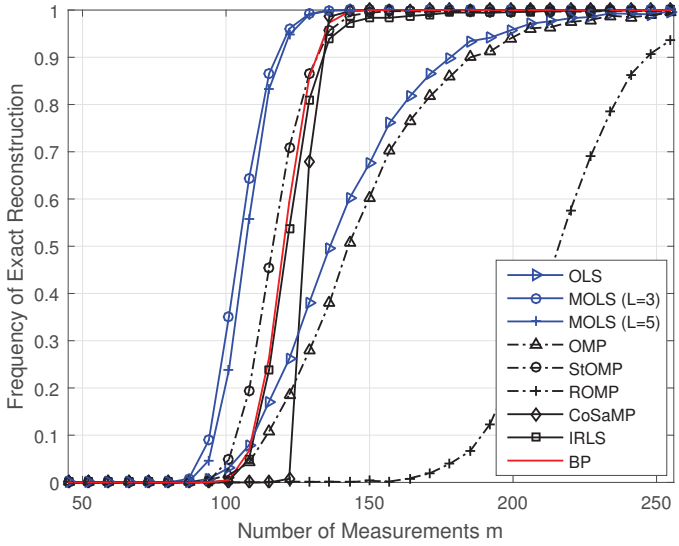}}
\caption{Frequency of exact recovery of sparse signals as a function of $m$.} \label{fig:varym}
\end{figure}

\begin{figure}[t]
\centering
\subfigure[$K=10$]
{\hspace{-2mm}
\includegraphics[width = 90 mm] {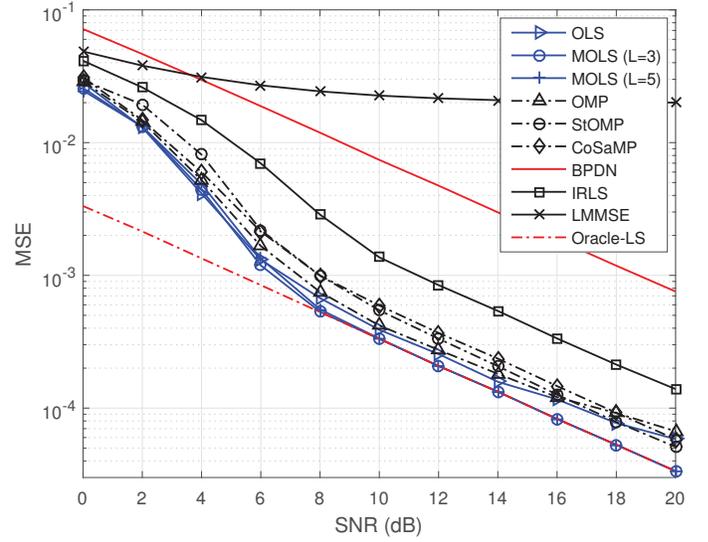}
\label{fig:subfig10}}
\subfigure[$K=20$]
{\hspace{-2mm}
\includegraphics[width = 90 mm] {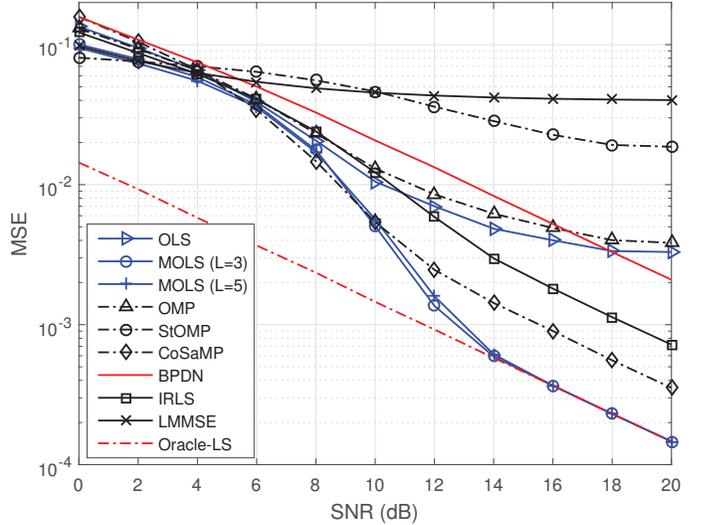}
\label{fig:subfig20}}
\caption{MSE performance of recovery methods for recovering sparse $2$-PAM signals as a function of SNR.}
\label{fig:mse1}\vspace{-1.5mm}
\end{figure}

In the noisy case, we empirically compare MSE performance of each recovery method. The MSE is defined as
\begin{equation}
\text{MSE}= \frac{1}{n} \sum_{i = 1}^n ({x}_i - \hat{{x}}_i)^2,
\end{equation}
where $\hat{{x}}_i$ is the estimate of $x_i$.
In obtaining the performance result for each simulation point of the algorithm, we perform $2,000$ independent trials.
In Fig.~\ref{fig:mse1}, we plot the MSE performance for each recovery method as a function of SNR (in dB) (i.e., $\text{SNR}: =  10 \log_{10} snr$).
In this case, the system model is expressed as $\mathbf{y} = \mathbf{\Phi x} + \mathbf{v}$ where $\mathbf{v}$ is the noise vector whose elements are generated from Gaussian distribution $\mathcal{N}(0, \frac{K}{m} 10^{- \frac{\text{SNR}}{10}})$.\footnote{Since the components of $\mathbf{\Phi}$ have power $\frac{1}{m}$ and the signal $\mathbf{x}$ is $K$-sparse with nonzero elements drawn independently from a standard Gaussian distribution, $\mathbb{E}|(\mathbf{\Phi x})_i|^2 = \frac{K}{m}$. From the definition of SNR, we have $\mathbb{E}|v_i|^2 = \mathbb{E}|(\mathbf{\Phi x})_i|^2 \cdot 10^{- \frac{\text{SNR}}{10}} = \frac{K}{m} 10^{- \frac{\text{SNR}}{10}}$.} The benchmark performance of Oracle least squares estimator (Oracle-LS), the best possible estimation having prior knowledge on the support of input signals, is plotted as well. In general, we observe that for all reconstruction methods, the MSE performance improves with the SNR.
For the whole SNR region under test, the MSE performance of MOLS is very competitive compared to that of other methods. In particular,  in the high SNR region the MSE of MOLS matches with that of the Oracle-LS estimator. This is essentially because the MOLS algorithm detects all support indices of sparse signals successfully in that SNR region. An interesting point we would like to mention
is that the actual recovery error of MOLS may be much smaller than
indicated in Theorem \ref{thm:noi}. Consider MOLS ($L = 5$) for
example. When SNR = $20$dB, $K = 20$, and ${v}_j \sim \mathcal{N}(0,
\frac{K}{m} 10^{- \frac{\text{SNR}}{10}})$, we have $\mathbb{E}
\|\mathbf{v}\|_2 = ({K} \cdot 10^{- \frac{\text{SNR}}{10}})^{1/2} =
\frac{\sqrt 5}{5}$. Thus, by assuming small isometry constants we
obtain from Theorem~\ref{thm:noi} that $\|\mathbf{x} -
\hat{\mathbf{x}}\|_2 \leq \frac{4 \sqrt 5}{5}$.\footnote{In this
case, we can verify that condition \eqref{eq:jjjjffffaa} in
Theorem~\ref{thm:noi} is fulfilled. To be specific, since sparse
$2$-PAM signals have $\kappa = 1$ and ${snr} =
10^{\frac{\text{SNR}}{10}} = 100$, when the measurement matrix
$\mathbf{\Phi}$ has small isometry constants, \eqref{eq:jjjjffffaa}
roughly becomes $\sqrt{100} \geq \frac{\sqrt 5 + 1}{\sqrt 5} \cdot
\sqrt{20}$, which is true.} Whereas, the $\ell_2$-norm of the actual
recovery error of MOLS is $\|\mathbf{x} - \hat{\mathbf{x}}\|_2 = (n
\cdot \text{MSE})^{1/2} \approx 0.2$ (Fig.~\ref{fig:subfig20}),
which is much smaller. The gap between the theoretical and empirical
results is perhaps due to 1) that our analysis is based on the RIP
framework and hence is essentially the worst-case-analysis, and 2)
that some inequalities (relaxations) used in our analysis may not be
tight.

\section{Conclusion} \label{sec:VII}

In this paper, we have studied a sparse recovery algorithm called
MOLS, which extends the conventional OLS algorithm by allowing
multiple candidates entering the list in each selection. Our method
is inspired by the fact that ``sub-optimal'' candidates in each of
the OLS identification are likely to be reliable and can be selected
to accelerate the convergence of the algorithm. We have demonstrated
by RIP analysis that MOLS ($L > 1$) performs exact recovery of any
$K$-sparse signal within $K$ iterations if $\delta_{LK} \leq
\frac{\sqrt L}{\sqrt K + 2 \sqrt{L}}$. In particular, for the
special case of MOLS when $L = 1$ (i.e., the OLS case), we have
shown that any $K$-sparse signal can be exactly recovered in $K$
iterations under $\delta_{K + 1} \leq \frac{1}{\sqrt K + 2}$, which
is a nearly optimal condition for the OLS algorithm. We have also
extended our analysis to the noisy scenario. Our result showed that
stable recovery of sparse signals can be achieved with MOLS when the
SNR has a linear scaling in the sparsity level of signals to be
recovered. In particular, for the case of OLS, we demonstrated from
a counterexample that the linear scaling law for the SNR is
essentially necessary. In addition, we have shown from empirical
experiments that the MOLS algorithm has lower computational cost
than the conventional OLS algorithm, while exhibiting improved
recovery accuracy. The empirical recovery performance of MOLS is
also competitive when compared to the state of the art recovery
methods.

\section*{Acknowledgement}
This work is  supported in part by ONR-N00014-13-1-0764, NSF-III-1360971, AFOSR-FA9550-13-1-0137, and NSF-Bigdata-1419210.

\appendices
\numberwithin{equation}{section}
\newcounter{mytempthcnt}
\setcounter{mytempthcnt}{\value{theorem}}
\setcounter{theorem}{2}

\section{Proof of Lemma~\ref{lem:rips2}} \label{app:new}

\begin{proof}
We focus on the proof for the upper bound. Since $\delta_{| \mathcal{S} |} <1$, $\mathbf{\Phi}_\mathcal{S}$ has full column
  rank. Suppose that $\mathbf{\Phi}_\mathcal{S}$ has singular value decomposition
  $\mathbf{\Phi}_\mathcal{S} = \mathbf{U} \mathbf{\Sigma} \mathbf{V}'$.
  Then from the definition of RIP, the minimum diagonal
  entries of $\mathbf{\Sigma}$ satisfies $\sigma_{\min} \geq \sqrt{1 - \delta_{| \mathcal{S}|}}.$ Note that
\begin{eqnarray}
  (\mathbf{\Phi}_\mathcal{S}^\dag)' & = & ((\mathbf{\Phi}'_\mathcal{S} \mathbf{\Phi}_\mathcal{S})^{-1} \mathbf{\Phi}'_\mathcal{S})' \nonumber \\
  & = & \mathbf{U} \mathbf{\Sigma} \mathbf{V}' ((\mathbf{U} \mathbf{\Sigma} \mathbf{V}')' \mathbf{U} \mathbf{\Sigma} \mathbf{V}')^{-1} = \mathbf{U} \mathbf{\Sigma}^{-1} \mathbf{V}',~~~~~~
\end{eqnarray}
  where
  $\mathbf{\Sigma}^{-1}$  is the diagonal matrix formed by replacing every (non-zero) diagonal
  entry of $\mathbf{\Sigma}$ by its reciprocal. Hence, all singular values
  of $(\mathbf{\Phi}_\mathcal{S}^\dag)'$ are upper bounded by
  $\frac{1}{\sigma_{\min}} = \frac{1}{\sqrt{1 - \delta_{| \mathcal{S} |}}}$, which competes the proof.
\end{proof}

\section{Proof of Proposition~\ref{prop:p1}} \label{app:p1}

\begin{proof}
Since $\mathbf{P}^{\bot}_{\mathcal{T}^{k} \cup \{i\}} \mathbf{y}$
and $\mathbf{P}_{\mathcal{T}^{k} \cup \{i\}} \mathbf{y}$ are
orthogonal,
\begin{equation}\notag
  \| \mathbf{P}^{\bot}_{\mathcal{T}^{k} \cup \{i\}} \mathbf{y} \|_{2}^{2}
  = \| \mathbf{y} \|_{2}^{2} - \| \mathbf{P}_{\mathcal{T}^{k} \cup \{i\}}
  \mathbf{y} \|_{2}^{2},
\end{equation}
and hence \eqref{eq:golsrulej} is equivalent to
\begin{equation}
  \label{eq:golsrule} \mathcal{S}^{k+1} = \arg \max_{\mathcal{S} : | \mathcal{S} | =L}   \sum_{i \in
  \mathcal{S}} \| \mathbf{P}_{\mathcal{T}^{k} \cup \{i\}} \mathbf{y} \|_{2}^{2}.
\end{equation}
By noting that $\mathbf{P}_{\mathcal{T}^{k}} +
  \mathbf{P}^{\bot}_{\mathcal{T}^{k}} = \mathbf{I}$, we have
\begin{eqnarray}
 \mathbf{P}_{\mathcal{T}^{k} \cup \{i\}}
  \hspace{-3mm}&=& \hspace{-3mm} \mathbf{P}_{\mathcal{T}^{k}} \mathbf{P}_{\mathcal{T}^{k} \cup \{i\}} +
  \mathbf{P}^{\bot}_{\mathcal{T}^{k}} \mathbf{P}_{\mathcal{T}^{k} \cup \{i\}}
  \nonumber\\
  \hspace{-3mm}&=& \hspace{-3mm} \mathbf{P}_{\mathcal{T}^{k}} +
  \mathbf{P}^{\bot}_{\mathcal{T}^{k}} \mathbf{P}_{\mathcal{T}^{k} \cup \{i\}}
  \nonumber \\
   \hspace{-3mm}&=& \hspace{-3mm} \mathbf{P}_{\mathcal{T}^{k}} \hspace{-.5mm}  + \hspace{-.5mm}
  \mathbf{P}^{\bot}_{\mathcal{T}^{k}} [\mathbf{\Phi}_{\mathcal{T}^k}, \phi_i] \hspace{-1mm} \left( \begin{bmatrix}
   \mathbf{\Phi}_{\mathcal{T}^k}' \\ \phi_i'
 \end{bmatrix}
 \hspace{-1mm}
 [\mathbf{\Phi}_{\mathcal{T}^k},\phi_i] \right)^{\hspace{-1mm} -1} \hspace{-1mm}
 \begin{bmatrix}
   \mathbf{\Phi}_{\mathcal{T}^k}' \\ \phi_i'
 \end{bmatrix}   \nonumber \\
  \hspace{-3mm}&=& \hspace{-3mm} \mathbf{P}_{\mathcal{T}^{k}} \hspace{-.5mm} + \hspace{-.5mm}
 \begin{bmatrix}
   \mathbf{0} \hspace{-1mm} & \mathbf{P}^{\bot}_{\mathcal{T}^{k}} \phi_i
 \end{bmatrix}\hspace{-1mm}
 \begin{bmatrix}
   \mathbf{\Phi}_{\mathcal{T}^k}' \mathbf{\Phi}_{\mathcal{T}^k} & \mathbf{\Phi}_{\mathcal{T}^k}' \phi_i \\\notag
   \phi_i' \mathbf{\Phi}_{\mathcal{T}^k} & \phi_i' \phi_i
 \end{bmatrix}^{\hspace{-.5mm}-1} \hspace{-1mm}
 \begin{bmatrix}
   \mathbf{\Phi}_{\mathcal{T}^k}' \\ \phi_i'
 \end{bmatrix}   \nonumber \\
 \hspace{-3mm} &\overset{(a)}{=}& \hspace{-3mm} \mathbf{P}_{\mathcal{T}^{k}} +
 \begin{bmatrix}
   \mathbf{0} \hspace{-1mm} & \mathbf{P}^{\bot}_{\mathcal{T}^{k}} \phi_i
 \end{bmatrix}
 \begin{bmatrix}
 \mathbf{M}_1 & \mathbf{M}_2 \\
 \mathbf{M}_3 & \mathbf{M}_4
 \end{bmatrix}
 \begin{bmatrix}
   \mathbf{\Phi}'_{\mathcal{T}^k} \\ \phi_i'
 \end{bmatrix}  \nonumber \\
   \hspace{-3mm}&=& \hspace{-3mm}
 \mathbf{P}_{\mathcal{T}^{k}} +
  \mathbf{P}^{\bot}_{\mathcal{T}^{k}} \phi_{i} ( \phi'_{i} \mathbf{P}^{\bot}_{\mathcal{T}^{k}}
  \phi_{i} )^{-1} \phi'_{i} \mathbf{P}^{\bot}_{\mathcal{T}^{k}},
\end{eqnarray}
where $(a)$ is from the partitioned inverse formula and
\begin{eqnarray}
  \mathbf{M}_1 &=& (\mathbf{\Phi}'_{\mathcal{T}^k} \mathbf{P}_i^\bot
  \mathbf{\Phi}_{\mathcal{T}^k})^{-1}, \nonumber \\
  \mathbf{M}_2 &=& -(\mathbf{\Phi}'_{\mathcal{T}^k} \mathbf{P}^\bot_i \mathbf{\Phi}_{\mathcal{T}^k})^{-1} \mathbf{\Phi}'_{\mathcal{T}^k} \phi_i (\phi_i'
  \phi_i)^{-1}, \nonumber \\
  \mathbf{M}_3 &=& -(\mathbf{\phi}'_{i} \mathbf{P}^\bot_{\mathcal{T}^k} \mathbf{\phi}_{i})^{-1} \mathbf{\phi}'_{i} \mathbf{\Phi}'_{\mathcal{T}^k} (\mathbf{\Phi}'_{\mathcal{T}^k}  \mathbf{\Phi}_{\mathcal{T}^k})^{-1}, \nonumber \\
  \mathbf{M}_4 &=& ({\phi}'_{i} \mathbf{P}_{\mathcal{T}^k}^\bot
  {\phi}_{i})^{-1}. \label{eq:48me}
\end{eqnarray}
This implies that
\begin{eqnarray}
 \| \mathbf{P}_{\mathcal{T}^{k} \cup \{i\}} \mathbf{y} \|_{2}^{2}
  \hspace{-3mm} &=& \hspace{-3mm} \| \mathbf{P}_{\mathcal{T}^{k}} \mathbf{y} + \mathbf{P}^{\bot}_{\mathcal{T}^{k}}
  \phi_{i} ( \phi'_{i} \mathbf{P}^{\bot}_{\mathcal{T}^{k}} \phi_{i} )^{-1} \phi'_{i}
  \mathbf{P}^{\bot}_{\mathcal{T}^{k}} \mathbf{y} \|_{2}^{2} \nonumber \\
 \hspace{-3mm} & \overset{(a)}{=} & \hspace{-3mm} \| \mathbf{P}_{\mathcal{T}^{k}} \mathbf{y} \|_{2}^{2} \hspace{-.75mm} + \hspace{-.75mm} \|
  \mathbf{P}^{\bot}_{\mathcal{T}^{k}} \phi_{i} ( \phi'_{i} \mathbf{P}^{\bot}_{\mathcal{T}^{k}}
  \phi_{i} )^{-1} \phi'_{i} \mathbf{P}^{\bot}_{\mathcal{T}^{k}} \mathbf{y}
  \|_{2}^{2} \nonumber\\
  \hspace{-3mm} & \overset{(b)}{=} & \hspace{-3mm} \| \mathbf{P}_{\mathcal{T}^{k}} \mathbf{y} \|_{2}^{2} + \frac{| \phi'_{i} \mathbf{P}^{\bot}_{\mathcal{T}^{k}} \mathbf{y} |^{2} \|
  \mathbf{P}^{\bot}_{\mathcal{T}^{k}} \phi_{i} \|^{2}_{2}}{|\phi'_{i} \mathbf{P}^{\bot}_{\mathcal{T}^{k}} \phi_{i}
  |^{2}} \nonumber\\
  \hspace{-3mm} & \overset{(c)}{=} & \hspace{-3mm} \| \mathbf{P}_{\mathcal{T}^{k}} \mathbf{y} \|_{2}^{2} + \frac{| \phi'_{i} \mathbf{P}^{\bot}_{\mathcal{T}^{k}} \mathbf{y} |^{2} }{\|
  \mathbf{P}^{\bot}_{\mathcal{T}^{k}} \phi_{i} \|^{2}_{2}} \nonumber \\
 \hspace{-3mm} & \overset{(d)}{=} & \hspace{-3mm} \| \mathbf{P}_{\mathcal{T}^{k}} \mathbf{y} \|_{2}^{2} + \left( \frac{| \langle \phi_{i}, \mathbf{r}^k
\rangle |}{\| \mathbf{P}^{\bot}_{\mathcal{T}^{k}} \phi_{i} \|_{2}}
\right)^2, \label{eq:decop2}
\end{eqnarray}
where (a) is because $\mathbf{P}_{\mathcal{T}^{k}} \mathbf{y}$ and
$\mathbf{P}^{\bot}_{\mathcal{T}^{k}} \phi_{i} ( \phi'_{i}
\mathbf{P}^{\bot}_{\mathcal{T}^{k}} \phi_{i} )^{-1} \phi'_{i}
\mathbf{P}^{\bot}_{\mathcal{T}^{k}} \mathbf{y}$ are orthogonal, (b)
follows from that fact that $\phi'_{i}
\mathbf{P}^{\bot}_{\mathcal{T}^{k}} \mathbf{y}$ and $\phi'_{i}
\mathbf{P}^{\bot}_{\mathcal{T}^{k}} \phi_i$ are scalars, (c) is from
\begin{equation}
\mathbf{P}^{\bot}_{\mathcal{T}^{k}} = (
\mathbf{P}^{\bot}_{\mathcal{T}^{k}} )^{2} = (
\mathbf{P}^{\bot}_{\mathcal{T}^{k}} )'\label{eq:Pbot}
\end{equation} and hence $|\phi'_{i}
\mathbf{P}^{\bot}_{\mathcal{T}^{k}} \phi_{i}| = |\phi'_{i}
(\mathbf{P}^{\bot}_{\mathcal{T}^{k}})'
\mathbf{P}^{\bot}_{\mathcal{T}^{k}} \phi_{i}| = \|
\mathbf{P}^{\bot}_{\mathcal{T}^{k}} \phi_{i} \|_2^2,$ and (d) is
due to $\mathbf{r}^k = \mathbf{P}^{\bot}_{\mathcal{T}^{k}}
\mathbf{y}$.

By relating \eqref{eq:golsrule} and {\eqref{eq:decop2}}, we have
\begin{equation}\notag
  \mathcal{S}^{k + 1} = \arg \max_{\mathcal{S} : | \mathcal{S} | =L}   \sum_{i \in \mathcal{S}}
\frac{| \langle \phi_{i}, \mathbf{r}^k \rangle |}{\|
\mathbf{P}^{\bot}_{\mathcal{T}^{k}} \phi_{i} \|_{2}}.
\end{equation}
Furthermore, if we write
$|\langle \phi_{i}, \mathbf{r}^k \rangle| = |\phi'_{i}
(\mathbf{P}^{\bot}_{\mathcal{T}^{k}})'
\mathbf{P}^{\bot}_{\mathcal{T}^{k}} \mathbf{y}| = |\langle
\mathbf{P}^{\bot}_{\mathcal{T}^{k}} \phi_{i}, \mathbf{r}^k
\rangle|,$ then \eqref{eq:golsrule11111} becomes
\begin{equation}\notag
  \mathcal{S}^{k + 1} = \arg \max_{\mathcal{S} : | \mathcal{S} | =L}   \sum_{i \in \mathcal{S}}
\bigg| \bigg\langle \frac{ \mathbf{P}^{\bot}_{\mathcal{T}^{k}}
\phi_{i} }{\| \mathbf{P}^{\bot}_{\mathcal{T}^{k}} \phi_{i} \|_{2}},
\mathbf{r}^k \bigg\rangle \bigg|.
\end{equation}
This completes the proof.
\end{proof}

\section{Proof of Proposition~\ref{lem:upperbound}}\label{app:upperbound}

\begin{proof} We first prove {\eqref{eq:small}} and then prove \eqref{eq:large}.

{\it 1) Proof of {\eqref{eq:small}}}:
    Since $u_{1}$ is the largest value of $\big \{\frac{| \langle \phi_{i},
    \mathbf{r}^{k} \rangle |}{\| \mathbf{P}^{\bot}_{\mathcal{T}^{k}} \phi_{i} \|_{2}} \big \}_{i \in \mathcal{T} \backslash \mathcal{T}^k}$, we have
    \begin{eqnarray}
u_{1}
     \hspace{-3mm} & {\geq} & \hspace{-3mm} \sqrt{\frac{1}{|\mathcal{T} \backslash \mathcal{T}^k |}  \sum_{i \in \mathcal{T} \backslash \mathcal{T}^k } \frac{| \langle \phi_{i}, \mathbf{r}^{k} \rangle |^2}{\|
      \mathbf{P}^{\bot}_{\mathcal{T}^{k}} \phi_{i} \|_{2}^2}} \nonumber\\
     \hspace{-3mm} & \overset{(a)}{\geq} & \hspace{-3mm} \frac{1}{\sqrt{|\mathcal{T} \backslash \mathcal{T}^k |}}  \sqrt{\sum_{i \in \mathcal{T} \backslash \mathcal{T}^k } |
      \langle \phi_{i}, \mathbf{r}^{k} \rangle |^2} \nonumber \\
     \hspace{-3mm} & {=} & \hspace{-3mm} \frac{\| \mathbf{\Phi}'_{\mathcal{T} \backslash \mathcal{T}^k }
      \mathbf{r}^{k} \|_{2} }{\sqrt{K - \ell}} = \frac{\| \mathbf{\Phi}'_{\mathcal{T} \backslash \mathcal{T}^k} \mathbf{P}^\bot_{\mathcal{T}^k} \mathbf{\Phi} \mathbf{x} \|_2}{\sqrt{K - \ell}},  \label{eq:nv}
    \end{eqnarray}
    where (a) is because $\|\mathbf{P}^\bot_{\mathcal{T}^k} \phi_i\|_2 \leq \|\phi_i\|_2 = 1$.
Observe that
\begin{eqnarray}
\| \mathbf{\Phi}'_{\mathcal{T} \backslash \mathcal{T}^k} \mathbf{P}^\bot_{\mathcal{T}^k} \mathbf{\Phi} \mathbf{x} \|_2
     \hspace{-3mm} & \overset{(a)}{=} & \hspace{-3mm}
  \| \mathbf{\Phi}'_{\mathcal{T} \backslash \mathcal{T}^k} \mathbf{P}^\bot_{\mathcal{T}^k} \mathbf{\Phi}_{\mathcal{T} \backslash \mathcal{T}^k} \mathbf{x}_{\mathcal{T} \backslash \mathcal{T}^k}\|_2 \nonumber \\
  \hspace{-3mm}& \overset{(b)}{\geq} &\hspace{-3mm}
  \frac{\|  \mathbf{x}'_{\mathcal{T} \backslash \mathcal{T}^k} \mathbf{\Phi}'_{\mathcal{T} \backslash \mathcal{T}^k} \mathbf{P}^\bot_{\mathcal{T}^k} \mathbf{\Phi}_{\mathcal{T} \backslash \mathcal{T}^k} \mathbf{x}_{\mathcal{T} \backslash \mathcal{T}^k}\|_2}{\| \mathbf{x}'_{\mathcal{T} \backslash \mathcal{T}^k}\|_2} \nonumber \\
 \hspace{-3mm} & \overset{(c)}{=} &\hspace{-3mm}
  \frac{\|  \mathbf{P}^\bot_{\mathcal{T}^k} \mathbf{\Phi}_{\mathcal{T} \backslash \mathcal{T}^k} \mathbf{x}_{\mathcal{T} \backslash \mathcal{T}^k}\|_2^2}{\| \mathbf{x}_{\mathcal{T} \backslash \mathcal{T}^k}\|_2} \nonumber \\
 \hspace{-3mm} & \overset{(d)}{\geq} & \hspace{-3mm} \lambda_{\min} ( \mathbf{\Phi}'_{\mathcal{T} \backslash \mathcal{T}^k} \mathbf{P}^\bot_{\mathcal{T}^k} \mathbf{\Phi}_{\mathcal{T} \backslash \mathcal{T}^k} )  \| \mathbf{x}_{\mathcal{T} \backslash \mathcal{T}^k} \|_{2} \nonumber \\
 \hspace{-3mm} & \overset{(e)}{\geq} & \hspace{-3mm}\lambda_{\min} ({ \mathbf{\Phi}'_{\mathcal{T}
      \cup \mathcal{T}^{k}} \mathbf{\Phi}_{\mathcal{T} \cup \mathcal{T}^{k}}}) \| \mathbf{x}_{\mathcal{T} \backslash \mathcal{T}^k} \|_{2} \nonumber\\
 \hspace{-3mm} & \overset{(f)}{\geq} & \hspace{-3mm}(1 - \delta_{K + Lk - \ell})\| \mathbf{x} _{\mathcal{T} \backslash \mathcal{T}^k} \|_{2},
      \label{eq:geaig1aa1}
    \end{eqnarray}
    where  (a) is because $\mathbf{P}^\bot_{\mathcal{T}^k} \mathbf{\Phi}_{\mathcal{T}^k} = \mathbf{0}$, (b) is from the norm inequality, (c) and (d) are from \eqref{eq:Pbot}, (e) is from Lemma~\ref{lem:rip6}, and (f) is from the RIP. (Note that $|\mathcal{T} \cup \mathcal{T}^{k}| = |\mathcal{T}| + |\mathcal{T}^{k}| - |\mathcal{T} \backslash \mathcal{T}^{k}| = K + Lk - \ell$.)

Using \eqref{eq:nv} and \eqref{eq:geaig1aa1}, we obtain {\eqref{eq:small}}.

\vspace{1mm}

{\it 2) Proof of {\eqref{eq:large}}}:
    Let $\mathcal{F}$ be the index set corresponding to $L$ largest elements in $\big \{\frac{|
    \langle \phi_{i}, \mathbf{r}^{k} \rangle |}{\| \mathbf{P}^{\bot}_{\mathcal{T}^{k}}
    \phi_{i} \|_{2}} \big\}_{i \in \Omega \setminus (\mathcal{T} \cup \mathcal{T}^{k} )}$. Then,
    \begin{eqnarray}
 \left( \sum_{i \in \mathcal{F}} \frac{| \langle \phi_{i},
      \mathbf{r}^{k} \rangle |^{2}}{\| \mathbf{P}^{\bot}_{\mathcal{T}^{k}} \phi_{i}
      \|_{2}^{2}} \right)^{\hspace{-1.5mm} 1/2}
      \hspace{-3mm}& \leq &\hspace{-3mm} \left( \frac{\sum_{i \in \mathcal{F}} | \langle \phi_{i}, \mathbf{r}^{k}
      \rangle |^{2}}{\min_{i \in \mathcal{F}} \|
      \mathbf{P}^{\bot}_{\mathcal{T}^{k}} \phi_{i} \|_{2}^{2}} \right)^{1/2} \nonumber\\
      \hspace{-3mm}& = &\hspace{-3mm} \left( \frac{\sum_{i \in \mathcal{F}} | \langle \phi_{i}, \mathbf{r}^{k}
      \rangle |^{2}}{1- \max_{i \in \mathcal{F}} \| \mathbf{P}_{\mathcal{T}^{k}}
      \phi_{i} \|_{2}^{2}} \right)^{1/2} \nonumber\\
      \nonumber\\
      \hspace{-3mm}& \overset{(a)}{\leq} &\hspace{-3mm} \left({1- \frac{\delta_{Lk+1}^{2}}{1-
      \delta_{Lk}}} \right)^{\hspace{-1mm} -1/2} \| \mathbf{\Phi}'_{\mathcal{F}} \mathbf{r}^{k} \|_{2},~~~~~~\label{eq:8700}
      \end{eqnarray}
  where (a) is because for any $i \notin \mathcal{T}^k$,
    \begin{eqnarray}
      \| \mathbf{P}_{\mathcal{T}^{k}}
      \phi_{i} \|_{2}^{2} &=&  \|( \mathbf{\Phi}_{\mathcal{T}^{k}}^{\dag} )'  \mathbf{\Phi}'_{\mathcal{T}^{k}} \phi_{i}
      \|_{2}^{2}  \nonumber \\
      & \overset{\text{Lemma}~\ref{lem:rips2}}{\leq} & \frac{\| \mathbf{\Phi}'_{\mathcal{T}^{k}}
      \phi_{i} \|_{2}^{2} }{1- \delta_{Lk}}  \nonumber\\
      & \overset{\text{Lemma}~\ref{lem:correlationrip} }{\leq} & \frac{\delta_{Lk+1}^{2}}{1- \delta_{Lk}}, \label{eq:C4}
    \end{eqnarray}
      By noting that $\mathbf{r}^{k} = \mathbf{y} - \mathbf{\Phi x}^{k} =
    \mathbf{y} - \mathbf{\Phi}_{\mathcal{T}^{k}}  \mathbf{\Phi}^{\dag}_{\mathcal{T}^{k}}
    \mathbf{y} = \mathbf{y} - \mathbf{P}_{\mathcal{T}^{k}} \mathbf{y} =
    \mathbf{P}^{\bot}_{\mathcal{T}^{k}} \mathbf{y} = \mathbf{P}^{\bot}_{\mathcal{T}^{k}}
    \mathbf{\Phi}_{\mathcal{T} \setminus \mathcal{T}^{k}} \mathbf{x}_{\mathcal{T} \setminus \mathcal{T}^{k}}$, we have
      \begin{eqnarray}
      &\lefteqn{\left( \sum_{i \in \mathcal{F}} \frac{| \langle \phi_{i},
      \mathbf{r}^{k} \rangle |^{2}}{\| \mathbf{P}^{\bot}_{\mathcal{T}^{k}} \phi_{i}
      \|_{2}^{2}} \right)^{1/2}} \nonumber \\
      & {=} & \left( \frac{1- \delta_{Lk}}{1- \delta_{Lk} -
      \delta_{Lk+1}^{2}} \right)^{1/2} \left\| \mathbf{\Phi}_{\mathcal{F}}'
      \mathbf{P}_{\mathcal{T}^{k}}^{\bot} \mathbf{\Phi}_{\mathcal{T} \setminus \mathcal{T}^{k}}
      \mathbf{x}_{\mathcal{T} \setminus \mathcal{T}^{k}} \right\|_{2} \nonumber\\
      & \leq & \left( \frac{1- \delta_{Lk}}{1- \delta_{Lk} -
      \delta_{Lk+1}^{2}} \right)^{1/2}  \left( \left\| \mathbf{\Phi}_{\mathcal{F}}'
      \mathbf{\Phi}_{\mathcal{T} \setminus \mathcal{T}^{k}} \mathbf{x}_{\mathcal{T} \setminus \mathcal{T}^{k}}
      \right\|_{2} \right.\nonumber \\
      & & \left. + \left\| \mathbf{\Phi}_{\mathcal{F}}' \mathbf{P}_{\mathcal{T}^{k}}
      \mathbf{\Phi}_{\mathcal{T} \setminus \mathcal{T}^{k}} \mathbf{x}_{\mathcal{T} \setminus \mathcal{T}^{k}}
      \right\|_{2} \right). \label{eq:54}
    \end{eqnarray}

     Since
    $\mathcal{F}$ and $\mathcal{T} \setminus \mathcal{T}^{k}$ are disjoint (i.e., $\mathcal{F} \cap (\mathcal{T} \setminus \mathcal{T}^{k}
    ) = \emptyset$), and also noting that $\mathcal{T}
    \cap \mathcal{T}^{k} = \ell$ by hypothesis, we have $|\mathcal{F}| + |\mathcal{T} \setminus \mathcal{T}^{k} | =L+K- \ell$. Using this together with Lemma~\ref{lem:correlationrip}, we have
\begin{equation}
          \left\| \mathbf{\Phi}_{\mathcal{F}}'  \mathbf{\Phi}_{\mathcal{T} \setminus \mathcal{T}^{k}}
      \mathbf{x}_{\mathcal{T} \setminus \mathcal{T}^{k}} \right\|_{2} \leq \delta_{L+K- \ell}
      \left\| \mathbf{x}_{\mathcal{T} \setminus \mathcal{T}^{k}} \right\|_{2}.  \label{eq:j1}
    \end{equation}
    Moreover, since $\mathcal{F} \cap \mathcal{T}^{k} = \emptyset$ and $|\mathcal{F}| + |\mathcal{T}^{k} | =L+Lk$,
    \begin{eqnarray}
      \lefteqn{\left\| \mathbf{\Phi}_{\mathcal{F}}' \mathbf{P}_{\mathcal{T}^{k}} \mathbf{\Phi}_{\mathcal{T} \setminus
      \mathcal{T}^{k}} \mathbf{x}_{\mathcal{T} \setminus \mathcal{T}^{k}} \right\|_{2}} \nonumber \\
      & \leq & \delta_{L+Lk}
      \left\| \mathbf{\Phi}_{\mathcal{T}^{k}}^{\dag}  \mathbf{\Phi}_{\mathcal{T} \setminus \mathcal{T}^{k}}
      \mathbf{x}_{\mathcal{T} \setminus \mathcal{T}^{k}} \right\|_{2} \label{eq:j2} \nonumber\\
      & = & \delta_{L+Lk}  \left\| ( \mathbf{\Phi}_{\mathcal{T}^{k}}'
      \mathbf{\Phi}_{\mathcal{T}^{k}} )^{-1}  \mathbf{\Phi}_{\mathcal{T}^{k}}'  \mathbf{\Phi}_{\mathcal{T}
      \setminus \mathcal{T}^{k}} \mathbf{x}_{\mathcal{T} \setminus \mathcal{T}^{k}} \right\|_{2} \nonumber\\
      & \overset{\text{Lemma}~\ref{lem:rips}}{\leq} & \label{eq:ghg1} \frac{\delta_{L+Lk}}{1-
      \delta_{Lk}}  \left\| \mathbf{\Phi}_{\mathcal{T}^{k}}'  \mathbf{\Phi}_{\mathcal{T} \setminus
      \mathcal{T}^{k}} \mathbf{x}_{\mathcal{T} \setminus \mathcal{T}^{k}} \right\|_{2} \nonumber\\
      & \overset{\text{Lemma}~\ref{lem:correlationrip}}{\leq} & \frac{\delta_{L+Lk} \delta_{Lk+K- \ell}}{1-
      \delta_{Lk}} \left\| \mathbf{x}_{\mathcal{T} \setminus \mathcal{T}^{k}} \right\|_{2}.
      \label{eq:ghg2}
    \end{eqnarray}
where in the last inequality we have used the fact that $|
\mathcal{T}^{k} \cup ( \mathcal{T} \setminus
    \mathcal{T}^{k} ) | =Lk+K- \ell$. (Note that $\mathcal{T}^{k}$ and
$\mathcal{T} \setminus \mathcal{T}^{k}$ are
    disjoint and $|\mathcal{T} \setminus \mathcal{T}^{k} | =K-
    \ell$.)

    Invoking (\ref{eq:j1}) and (\ref{eq:ghg2}) into (\ref{eq:54}), we have
    \begin{eqnarray}
      \lefteqn{\left( \sum_{i \in \mathcal{F}} \frac{| \langle \phi_{i}, \mathbf{r}^{k} \rangle
      |^{2}}{\| \mathbf{P}^{\bot}_{\mathcal{T}^{k}} \phi_{i} \|_{2}^{2}}
      \right)^{1/2} \leq  \left( \frac{1- \delta_{Lk}}{1- \delta_{Lk} - \delta_{Lk+1}^{2}}
      \right)^{1/2} } \nonumber \\
      & & \times \left( \delta_{L+K- \ell} + \frac{\delta_{L+Lk}
      \delta_{Lk+K- \ell}}{1- \delta_{Lk}} \right) \left\| \mathbf{x} _{\mathcal{T} \setminus \mathcal{T}^{k}} \right\|_{2}.~~~~~  \label{eq:left}
    \end{eqnarray}
    On the other hand, since $v_{L}$ is the $L$-th largest value in
    $\big \{\frac{| \langle \phi_{i}, \mathbf{r}^{k} \rangle |}{\|
    \mathbf{P}^{\bot}_{\mathcal{T}^{k}} \phi_{i} \|_{2}} \big \}_{i \in \mathcal{F}}$, we have
    \begin{eqnarray}
      \left( \sum_{i \in \mathcal{F}} \frac{| \langle \phi_{i}, \mathbf{r}^{k} \rangle
      |^{2}}{\| \mathbf{P}^{\bot}_{\mathcal{T}^{k}} \phi_{i} \|_{2}^{2}} \right)^{1/2}
      \geq \sqrt{L} v_{L},  \label{eq:right00}
    \end{eqnarray}
which, together with {\eqref{eq:left}}, implies \eqref{eq:large}.

\end{proof}

\section{Proof of \eqref{eq:k+1}} \label{app:cond}

\begin{proof}
Observe that \eqref{eq:sufficientommp4} is equivalent to
  \begin{eqnarray}
    \label{eq:sufficientommp5}
    \sqrt{\frac{K - \ell}{L}} < \frac{(1-\delta_{LK})^{3/2}( 1- \delta_{LK} - \delta_{LK}^{2}
    )^{1/2}}{\delta_{LK}}.
  \end{eqnarray}
Let $f(\delta_{LK}) = \frac{(1-\delta_{LK})^{3/2}( 1- \delta_{LK} - \delta_{LK}^{2}
    )^{1/2}}{\delta_{LK}}$ and
$g(\delta_{LK}) = \frac{1}{\delta_{LK}} - 2.$
Then one can check that $\forall \delta_{LK} \in (0,\frac{\sqrt 5 - 1}{2})$,
\begin{equation}\label{eq:fggg}
f(\delta_{LK}) > g(\delta_{LK}).
\end{equation}
Hence, \eqref{eq:sufficientommp5}
is ensured by $
    \sqrt{\frac{K - \ell}{L}} < \frac{1}{\delta_{LK}} - 2$,
or equivalently,
  \begin{equation}
    \delta_{LK} < \frac{\sqrt{L}}{\sqrt{K - \ell} + 2 \sqrt{L}}.
    \label{eq:sufficientommp31}
  \end{equation}
Since $K - \ell < K$, \eqref{eq:sufficientommp31} is guaranteed by
\eqref{eq:k+1}.
\end{proof}

\section{Proof of Theorem~\ref{thm:8}}
\label{app:8}
\begin{proof}
We prove Theorem~\ref{thm:8} in two steps. First, we show that the residual
power difference of MOLS satisfies
\begin{equation} \label{eq:residual_A111}
  \hspace{-0mm}\|\mathbf{r}^k\|_2^2 \hspace{-.5mm} - \hspace{-.5mm} \|\mathbf{r}^{k + 1}\|_2^2 \hspace{-.5mm} \geq \hspace{-.5mm} \frac{1- \delta_{Lk} - \delta_{Lk+1}^{2}}{(1 + \delta_{L}) (1- \delta_{Lk})}  \max_{\mathcal{S} : | \mathcal{S} |
  =L} \hspace{-.5mm} \| \mathbf{\Phi}'_{\mathcal{S}} \mathbf{r}^{k}
  \|_{2}^{2}. \hspace{-2mm}
\end{equation}
In the second step, we show that
  \begin{equation} \label{eq:a111}
\max_{\mathcal{S} : | \mathcal{S} |
  =L} \| \mathbf{\Phi}'_{\mathcal{S}} \mathbf{r}^{k} \|_{2}^{2} \geq \frac{L (1 - \delta_{K + Lk})^2}{K (1 + \delta_{K + Lk})}  \|\mathbf{r}^k \|_{2}^{2}.
  \end{equation}
The theorem is established by combining \eqref{eq:residual_A111}
and \eqref{eq:a111}.

{1) \it Proof of \eqref{eq:residual_A111}}: First, from the definition of MOLS (see Table~\ref{tab:mols}), we have that for any integer $0 \leq k < K$,
\begin{eqnarray}
 \mathbf{r}^k - \mathbf{r}^{k + 1}
&\overset{(a)}{=}& \mathbf{P}_{\mathcal{T}^{k}} \mathbf{y} -
\mathbf{P}_{\mathcal{T}^{k + 1}} \mathbf{y} \nonumber \\
&\overset{(b)}{=}& (\mathbf{P}_{\mathcal{T}^{l + 1}} - \mathbf{P}_{\mathcal{T}^{l + 1}}
\mathbf{P}_{\mathcal{T}^{l}} )\mathbf{y} \nonumber \\ &=& \mathbf{P}_{\mathcal{T}^{l +
1}} (\mathbf{y} - \mathbf{P}_{\mathcal{T}^{l}}
\mathbf{y}) = \mathbf{P}_{\mathcal{T}^{l + 1}} \mathbf{r}^k,
\end{eqnarray}
where (a) is from that $\mathbf{x}^{k} = \underset{\mathbf{u}:\textit{supp}(\mathbf{u}) = \mathcal{T}^k}{\arg \min} \|\mathbf{y}-\mathbf{\Phi} \mathbf{u}\|_2$, and hence $\mathbf{\Phi} \mathbf{x}^{k} =  \mathbf{\Phi}_{\mathcal{T}^k} \mathbf{\Phi}^\dag_{\mathcal{T}^k} \mathbf{y} = \mathbf{P}_{\mathcal{T}^{l}} \mathbf{y}$, and (b) is because $\text{span}(\mathbf{\Phi}_{\mathcal{T}^{k}})
\subseteq \text{span}(\mathbf{\Phi}_{\mathcal{T}^{k + 1}})$ so that
$\mathbf{P}_{\mathcal{T}^k} \mathbf{y} = \mathbf{P}_{\mathcal{T}^{k + 1}}
(\mathbf{P}_{\mathcal{T}^{k}} \mathbf{y})$.
Since $\mathcal{T}^{k + 1} \supseteq \mathcal{S}^{k + 1}$, we have $\text{span}(\mathbf{\Phi}_{\mathcal{T}^{k + 1}}) \supseteq
\text{span}({\mathbf{\Phi}}_{S\mathcal{}^{k + 1}})$ and
\begin{eqnarray} \label{eq:residual_A}
  \|\mathbf{r}^k - \mathbf{r}^{k + 1}\|_2^2 = \|\mathbf{P}_{\mathcal{T}^{k + 1}}
 \mathbf{r}^k\|_2^2 \geq \|\mathbf{P}_{\mathcal{S}^{k + 1}}
 \mathbf{r}^k\|_2^2. \nonumber
\end{eqnarray}
By noting that $\|\mathbf{r}^k - \mathbf{r}^{k + 1}\|_2^2 =
\|\mathbf{r}^k\|_2^2 - \|\mathbf{r}^{k + 1}\|_2^2$, we have
\begin{eqnarray}
  \lefteqn{\|\mathbf{r}^k \|_2^2 - \|\mathbf{r}^{k + 1}\|_2^2} \nonumber \\
  &\geq& \|\mathbf{P}_{\mathcal{S}^{k + 1}}
  \mathbf{r}^k\|_2^2 \nonumber \\
&\overset{(a)}{\geq}& \|(\mathbf{\Phi}_{\mathcal{S}^{k + 1}}^\dag)' \mathbf{\Phi}'_{\mathcal{S}^{k + 1}} \mathbf{r}^k\|_2^2 \nonumber \\
&\overset{\text{Lemma}~\ref{lem:rips2}}{\geq}& \frac{\|\mathbf{\Phi}'_{\mathcal{S}^{k + 1}} \mathbf{r}^k\|_2^2}{1 + \delta_{|\mathcal{S}^{k + 1}|}} = \frac{\|\mathbf{\Phi}'_{\mathcal{S}^{k + 1}} \mathbf{r}^k\|_2^2}{1 + \delta_{L}},~~~~~~~
 \label{eq:mmmss4}
\end{eqnarray}
where (a) is because $\mathbf{P}_{\mathcal{S}^{k + 1}} = \mathbf{P}'_{\mathcal{S}^{k + 1}} = (
\mathbf{\Phi}_{\mathcal{S}^{k + 1}}^\dag)' \mathbf{\Phi}'_{\mathcal{S}^{k + 1}}$.

Next, we build a lower bound for $\|\mathbf{\Phi}'_{\mathcal{S}^{k + 1}} \mathbf{r}^k\|_2^2$. Denote $S^* = \arg \max_{\mathcal{S} : | \mathcal{S} |  =L} \| \mathbf{\Phi}'_{\mathcal{S}} \mathbf{r}^{k} \|_{2}^{2}$. Then,
\begin{eqnarray}
  \sum_{i \in \mathcal{S}^{k + 1}} \frac{| \langle \phi_{i}, \mathbf{r}^{k} \rangle
  |^{2}}{\| \mathbf{P}^{\bot}_{\mathcal{T}^{k}} \phi_{i} \|_{2}^{2}}
  &\hspace{-3mm} \overset{{\eqref{eq:golsrule11111}}}{=} &\hspace{-3mm} \max_{\mathcal{S} : | \mathcal{S} | =L}   \sum_{i \in \mathcal{S}}
\frac{| \langle \phi_{i}, \mathbf{r}^k \rangle |^2}{\|
\mathbf{P}^{\bot}_{\mathcal{T}^{k}} \phi_{i} \|_{2}^2} \nonumber \\
  &\hspace{-3mm}\geq & \hspace{-3mm} \sum_{i
  \in {\mathcal{S}}^{*}} \frac{| \langle \phi_{i}, \mathbf{r}^{k}
  \rangle |^{2}}{\| \mathbf{P}^{\bot}_{\mathcal{T}^{k}} \phi_{i} \|_{2}^{2}}
  \label{eq:comparewithgomp} \label{eq:comparewithgomp1} \nonumber \\
  &\hspace{-3mm} \overset{(a)}{\geq} & \hspace{-3mm} \sum_{i \in {\mathcal{S}}^{*}} | \langle \phi_{i},
  \mathbf{r}^{k} \rangle |^{2} \hspace{-1mm} = \hspace{-1mm} \max_{\mathcal{S} : | \mathcal{S} |  =L} \| \mathbf{\Phi}'_{\mathcal{S}} \mathbf{r}^{k} \|_{2}^{2},~~~~~~ \label{eq:righteq}
\end{eqnarray}
where (a) holds because $\phi_{i}$ has
unit $\ell_{2}$-norm and hence $\|
\mathbf{P}^{\bot}_{\mathcal{T}^{k}} \phi_{i} \|_{2} \leq 1$.

On the other hand,
\begin{eqnarray}
 \sum_{i \in \mathcal{S}^{k + 1}} \frac{| \langle \phi_{i}, \mathbf{r}^{k} \rangle|^{2}}{\| \mathbf{P}^{\bot}_{\mathcal{T}^{k}} \phi_{i} \|_{2}^{2}}
  &\hspace{-3mm}\leq & \hspace{-3mm} \frac{\sum_{i \in \mathcal{S}^{k + 1}} \hspace{-2mm} | \langle \phi_{i}, \mathbf{r}^{k} \rangle |^{2} }{\min_{i \in \mathcal{S}^{k + 1}} \|\mathbf{P}^{\bot}_{\mathcal{T}^{k}} \phi_{i} \|_{2}^{2}} \nonumber \\
  &\hspace{-3mm}= & \hspace{-3mm} \frac{\sum_{i \in \mathcal{S}^{k + 1}} \hspace{-2mm} | \langle \phi_{i}, \mathbf{r}^{k} \rangle |^{2}}{1\hspace{-.5mm} - \hspace{-.5mm} \max_{i \in \mathcal{S}^{k + 1}} \hspace{-1mm} \| \mathbf{P}_{\mathcal{T}^{k}} \phi_{i} \|_{2}^{2}} \nonumber \\
  &\hspace{-3mm}\overset{\eqref{eq:C4}}{\leq} & \hspace{-3mm} \left({1- \frac{\delta_{Lk+1}^{2}}{1- \delta_{Lk}}}
  \right)^{\hspace{-1mm}-1} \hspace{-1mm} \| \mathbf{\Phi}'_{\mathcal{S}^{k + 1}} \mathbf{r}^{k} \|_{2}^{2}.~~~~~ \label{eq:lefteq}
\end{eqnarray}
Combining {\eqref{eq:righteq}} and {\eqref{eq:lefteq}} yields
\begin{eqnarray}
  \| \mathbf{\Phi}'_{\mathcal{S}^{k + 1}} \mathbf{r}^{k} \|_{2}^{2}
  \geq \left( 1-
  \frac{\delta_{Lk+1}^{2}}{1- \delta_{Lk}} \right)
  \max_{\mathcal{S} : | \mathcal{S} |  =L} \| \mathbf{\Phi}'_{\mathcal{S}} \mathbf{r}^{k} \|_{2}^{2}. \label{eq:45}
\end{eqnarray}

Finally, using \eqref{eq:mmmss4} and \eqref{eq:45}, we obtain \eqref{eq:residual_A111}.

\vspace{1mm}

{\it 2) Proof of \eqref{eq:a111}}: Since $L \leq K$,
    \begin{eqnarray}
 \max_{\mathcal{S} : | \mathcal{S} |
  =L} \| \mathbf{\Phi}'_{\mathcal{S}} \mathbf{r}^{k} \|_{2}^{2}
  &\hspace{-2.5mm} {\geq} &\hspace{-3mm} \frac{L}{K}\|\mathbf{\Phi}'_{\mathcal{T}} \mathbf{r}^{k}\|_{2}^{2} \overset{(a)}{=}  \frac{L}{K}\| \mathbf{\Phi}'_{\mathcal{T} \cup \mathcal{T}^k}  \mathbf{r}^{k} \|_{2}^{2} \nonumber \\
  &\hspace{-3mm} = &\hspace{-3mm} \frac{L}{K} \| \mathbf{\Phi}'_{\mathcal{T} \cup \mathcal{T}^k} \mathbf{\Phi}_{\mathcal{T} \cup \mathcal{T}^k} (\mathbf{x} - \mathbf{x}^{k})_{\mathcal{T} \cup \mathcal{T}^k} \|_{2}^{2} \nonumber \\
   &\hspace{-3mm} \overset{\text{RIP}}{\geq} &\hspace{-3mm} \frac{L}{K} (1 - \delta_{K + Lk})^2 \|(\mathbf{x} - \mathbf{x}^{k})_{\mathcal{T} \cup \mathcal{T}^k} \|_{2}^{2} \nonumber \\
  &\hspace{-3mm} \overset{\text{RIP}}{\geq} &\hspace{-3mm} \frac{L (1 - \delta_{K + Lk})^2}{K (1 + \delta_{K + Lk})}  \|\mathbf{\Phi}_{\mathcal{T} \cup \mathcal{T}^k} (\mathbf{x} \hspace{-.51mm} - \hspace{-.51mm} \mathbf{x}^{k})_{\mathcal{T} \cup \mathcal{T}^k} \|_{2}^{2} \nonumber \hspace{-1mm} \\
  &\hspace{-3mm} = &\hspace{-3mm} \frac{L (1 - \delta_{K + Lk})^2}{K (1 + \delta_{K + Lk})}  \|\mathbf{r}^k \|_{2}^{2}, \label{eq:51j}
    \end{eqnarray}
where (a) is because $\mathbf{\Phi}'_{\mathcal{T}^k} \mathbf{r}^k = \mathbf{\Phi}'_{\mathcal{T}^k} (\mathbf{P}^\bot_{\mathcal{T}^k} \mathbf{y}) = \mathbf{\Phi}'_{\mathcal{T}^k} (\mathbf{P}^\bot_{\mathcal{T}^k})' \mathbf{y} = (\mathbf{P}^\bot_{\mathcal{T}^k} \mathbf{\Phi}_{\mathcal{T}^k})' \mathbf{y} = \mathbf{0}$.

From \eqref{eq:mmmss4} and \eqref{eq:51j},
\begin{eqnarray}
 \|\mathbf{r}^k\|_2^2 \hspace{-.75mm} - \hspace{-.75mm} \|\mathbf{r}^{k + 1} \hspace{-.5mm}\|_2^2 \hspace{-.5mm} \geq \hspace{-1mm} \frac{L (1 - \delta_{K + Lk})^2}{K \hspace{-.25mm}(1 \hspace{-.5mm} + \hspace{-.5mm} \delta_{L}) (1 \hspace{-.5mm} + \hspace{-.5mm} \delta_{K + Lk})} \hspace{-.5mm} \left(\hspace{-.5mm} 1 \hspace{-.5mm} - \hspace{-.5mm} \frac{\delta_{Lk+1}^{2}}{1 - \delta_{Lk}} \hspace{-.5mm} \right) \hspace{-.5mm} \|\mathbf{r}^k \hspace{-.5mm} \|_{2}^{2}, \nonumber
\end{eqnarray}
which implies that $
\|\mathbf{r}^{k + 1}\|_2^2 \leq \alpha(k, L) \|\mathbf{r}^{k}\|_2^2
$
where
\begin{equation}
\alpha(k, L) := 1 - \frac{L (1- \delta_{Lk} - \delta_{Lk+1}^{2}) (1 - \delta_{K + Lk})^2}{K (1 + \delta_{L}) (1- \delta_{Lk}) (1 + \delta_{K + Lk})}.
\end{equation}
Repeating this we obtain
\begin{equation}
\|\mathbf{r}^{k + 1}\|_2^2 \leq \prod_{i = 0}^{k + 1} \alpha(i, L) \|\mathbf{r}^{0}\|_2^2 \leq
  (\alpha(k, L))^{k + 1} \|\mathbf{y}\|_2^2,
\end{equation}
which completes the proof.
\end{proof}

\section{Proof of Theorem~\ref{thm:noi1}}\label{app:noi1}

\begin{proof}
We consider the best $K$-term approximation $(\mathbf{x}^{l})_{\hat{\mathcal{T}}}$ of ${\mathbf{x}}^{l}$ and observer that
\begin{eqnarray}
 \|(\mathbf{x}^{l})_{\hat{\mathcal{T}}} - \mathbf{x}\|_2
&=& \|(\mathbf{x}^{l})_{\hat{\mathcal{T}}} - \mathbf{x}^{l} + \mathbf{x}^{l} - \mathbf{x}\|_2
\nonumber \\
&\overset{(a)}{\leq}& \|(\mathbf{x}^{l})_{\hat{\mathcal{T}}} - \mathbf{x}^{l}\|_2 + \|\mathbf{x}^{l} - \mathbf{x}\|_2   \nonumber \\
& \overset{(b)}{\leq} & 2 \|\mathbf{x}^{l} - \mathbf{x}\|_2 \nonumber \\
& \overset{\text{RIP}}{\leq} & \frac{2 \|\mathbf{\Phi}(\mathbf{x}^{l} - \mathbf{x})\|_2}{\sqrt{1 - \delta_{Ll + K}}} = \frac{2 (\|\mathbf{r}^{l}\|_2 + \|\mathbf{v}\|_2)}{\sqrt{1 - \delta_{Ll + K}}} \nonumber \\
&\leq& \frac{2(\epsilon + \|\mathbf{v}\|_2)}{\sqrt{1 - \delta_{Ll + K}}}, \label{eq:zuuuo1}
\end{eqnarray}
where (a) is from the triangle inequality and (b) is because $(\mathbf{x}^{l})_{\hat{\mathcal{T}}}$ is the best $K$-term approximation to $\mathbf{x}^{l}$ and hence is a better approximation than $\mathbf{x}$.

On the other hand,
\begin{eqnarray}
\lefteqn{\|(\mathbf{x}^{l})_{\hat{\mathcal{T}}} - \mathbf{x}\|_2} \nonumber \\
&\overset{\text{RIP}}{\geq}& \frac{\|\mathbf{\Phi}((\mathbf{x}^{l})_{\hat{\mathcal{T}}} - \mathbf{x})\|_2}{\sqrt{1 - \delta_{2K}}} = \frac{\|\mathbf{\Phi}(\mathbf{x}^{l})_{\hat{\mathcal{T}}} - \mathbf{y} + \mathbf{v}\|_2}{\sqrt{1 - \delta_{2K}}} \nonumber \\
&\overset{(a)}{\geq}& \frac{\|\mathbf{\Phi}(\mathbf{x}^{l})_{\hat{\mathcal{T}}} - \mathbf{y} \|_2 - \| \mathbf{v}\|_2}{\sqrt{1 - \delta_{2K}}} \nonumber \\
&\overset{(b)}{\geq}& \frac{\|\mathbf{\Phi}\hat{\mathbf{x}} - \mathbf{y} \|_2 - \| \mathbf{v}\|_2}{\sqrt{1 - \delta_{2K}}} = \frac{\|\mathbf{\Phi}(\hat{\mathbf{x}} - \mathbf{x}) - \mathbf{v} \|_2 - \| \mathbf{v}\|_2}{\sqrt{1 - \delta_{2K}}} \nonumber \\
&\overset{(c)}{\geq}& \frac{\|\mathbf{\Phi}(\hat{\mathbf{x}} - \mathbf{x})\|_2 - 2 \| \mathbf{v}\|_2}{\sqrt{1 - \delta_{2K}}} \nonumber \\
&\overset{\text{RIP}}{\geq}& \frac{\sqrt{1 + \delta_{2K}} \|\hat{\mathbf{x}} - \mathbf{x}\|_2 - 2 \| \mathbf{v}\|_2}{\sqrt{1 - \delta_{2K}}},   \label{eq:79eq1}
\end{eqnarray}
where (a) and (c) are from the triangle inequality and (b) is because $(\mathbf{x}^{l})_{\hat{\mathcal{T}}}$ is supported on $\hat{\mathcal{T}}$ and $\hat{\mathbf{x}}_{\hat{\mathcal{T}}} = \mathbf{\Phi}^\dag_{\hat{\mathcal{T}}} \mathbf{y} = \arg
\underset{\mathbf{u}}{ \min} \|\mathbf{y} - \mathbf{\Phi}_{\hat{\mathcal{T}}} \mathbf{u}\|_2$ (see Table~\ref{tab:mols}).

By combining \eqref{eq:zuuuo1} and \eqref{eq:79eq1} we obtain~\eqref{eq:lx}.
\end{proof}

\section{Proof of Proposition~\ref{prop:r2}}\label{app:r2}

 \begin{proof}
We first consider the case of $L = 1$. In this case, $\mathcal{T}^K = \mathcal{T}$ and $\hat{\mathbf{x}} = \mathbf{x}^K = \mathbf{\Phi}^\dag_{\mathcal{T}} \mathbf{y}$, and hence
\begin{eqnarray}
\|\hat{\mathbf{x}} - \mathbf{x}\|_2 &=& \|\mathbf{x} - \mathbf{\Phi}^\dag_{\mathcal{T}} \mathbf{y}\|_2 = \|\mathbf{\Phi}^\dag_{\mathcal{T}} \mathbf{v}\|_2 \nonumber \\
&\overset{\text{RIP}}{\leq}& \frac{ \|\mathbf{\Phi}_{\mathcal{T}} \mathbf{\Phi}^\dag_{\mathcal{T}} \mathbf{v}\|_2}{\sqrt{1 - \delta_{K}}}
= \frac{\|\mathbf{P}_{\mathcal{T}}  \mathbf{v}\|_2}{\sqrt{1 - \delta_{K}}} \nonumber \\
 &\leq&  \frac{\|\mathbf{v}\|_2}{\sqrt{1 - \delta_{K}}}.
\end{eqnarray}

Next, we prove the case of $L > 1$. Consider the best $K$-term approximation $(\mathbf{x}^{K})_{\hat{\mathcal{T}}}$ of ${\mathbf{x}}^{K}$ and observer that
\begin{eqnarray}
 \|(\mathbf{x}^{K})_{\hat{\mathcal{T}}} - \mathbf{x}\|_2
&=& \|(\mathbf{x}^{K})_{\hat{\mathcal{T}}} - \mathbf{x}^{K} + \mathbf{x}^{K} - \mathbf{x}\|_2
\nonumber \\
&\overset{(a)}{\leq}& \|(\mathbf{x}^{K})_{\hat{\mathcal{T}}} - \mathbf{x}^{K}\|_2 + \|\mathbf{x}^{K} - \mathbf{x}\|_2   \nonumber \\
& \overset{(b)}{\leq} & 2 \|\mathbf{x}^{K} - \mathbf{x}\|_2 = 2 \|\mathbf{x} - \mathbf{\Phi}^\dag_{\mathcal{T}^K} \mathbf{y}\|_2 \nonumber \\
& \overset{(c)}{=} & 2 \|\mathbf{\Phi}^\dag_{\mathcal{T}^K} \mathbf{v}\|_2  \nonumber \\
&\overset{\text{RIP}}{\leq}&  \frac{2 \|\mathbf{\Phi}_{\mathcal{T}^K} \mathbf{\Phi}^\dag_{\mathcal{T}^K} \mathbf{v}\|_2}{\sqrt{1 - \delta_{|\mathcal{T}^K|}}} \nonumber \\
&=& \frac{2 \|\mathbf{P}_{\mathcal{T}^K}  \mathbf{v}\|_2}{\sqrt{1 - \delta_{LK}}}  \leq  \frac{2 \|\mathbf{v}\|_2}{\sqrt{1 - \delta_{LK}}}, \label{eq:zuuuo}
\end{eqnarray}
where (a) is from the triangle inequality, (b) is because $(\mathbf{x}^{K})_{\hat{\mathcal{T}}}$ is the best $K$-term approximation to $\mathbf{x}^{K}$ and hence is a better approximation than $\mathbf{x}$ (note that both $(\mathbf{x}^{K})_{\hat{\mathcal{T}}}$ and $\mathbf{x}$ are $K$-sparse), and (c) is because $\mathcal{T}^K \supseteq \mathcal{T}$ and $\mathbf{y} = \mathbf{\Phi x} + \mathbf{v}$.

On the other hand, following the same argument in \eqref{eq:79eq1}, one can show that
\begin{equation}
 \|(\mathbf{x}^{K})_{\hat{\mathcal{T}}} - \mathbf{x}\|_2 \geq \frac{\sqrt{1 + \delta_{2K}} \|\hat{\mathbf{x}} - \mathbf{x}\|_2 - 2 \| \mathbf{v}\|_2}{\sqrt{1 - \delta_{2K}}}.   \label{eq:79eq}
\end{equation}

Combining \eqref{eq:zuuuo} and \eqref{eq:79eq} yields the desired result.
\end{proof}

\section{Proof of Proposition~\ref{prop:upperbound15}}\label{app:upperbound1}

\begin{proof} In the following, we provide the proofs of \eqref{eq:small5} and \eqref{eq:large5},  respectively.

{\it 1) Proof of {\eqref{eq:small5}}}:
    Since $u'_{1}$ is the largest value of $\big \{\frac{| \langle \phi_{i}, \mathbf{r}^{k} \rangle |}{\| \mathbf{P}^{\bot}_{\mathcal{T}^{k}} \phi_{i} \|_{2}} \big \}_{i \in \mathcal{T} \backslash \mathcal{T}^k}$, we have
    \begin{eqnarray}
u'_{1}
     \hspace{-2mm} & \overset{(a)}{\geq} & \hspace{-3mm} \sqrt{\frac{1}{|\mathcal{T}  \backslash \mathcal{T}^k |}  \sum_{i \in \mathcal{T}  \backslash \mathcal{T}^k} \frac{ \langle \phi_{i}, \mathbf{r}^{k} \rangle ^2}{\|
      \mathbf{P}^{\bot}_{\mathcal{T}^{k}} \phi_{i} \|_{2}^2}} \geq \sqrt{\frac{\sum_{i \in \mathcal{T} \backslash \mathcal{T}^k }
      \langle \phi_{i}, \mathbf{r}^{k} \rangle ^2}{|\mathcal{T}  \backslash \mathcal{T}^k |}} \nonumber \\
      \hspace{-3mm} & = & \hspace{-3mm}
      \frac{\| \mathbf{\Phi}'_{\mathcal{T}  \backslash \mathcal{T}^k }
      \mathbf{r}^{k} \|_{2}}{\sqrt{K - \ell'}}
     =
     \frac{\| \mathbf{\Phi}'_{\mathcal{T}  \backslash \mathcal{T}^k }
      \mathbf{P}^\bot_{\mathcal{T}^k} (\mathbf{\Phi x} + \mathbf{v}) \|_{2}}{\sqrt{K - \ell'}}  \nonumber\\
     \hspace{-3mm} & \overset{(b)}{\geq} & \hspace{-3mm} \frac{\| \mathbf{\Phi}'_{\mathcal{T}  \backslash \mathcal{T}^k } \mathbf{P}^\bot_{\mathcal{T}^k}  \mathbf{\Phi}  \mathbf{x} \|_2 - \| \mathbf{\Phi}'_{\mathcal{T}  \backslash \mathcal{T}^k } \mathbf{P}^\bot_{\mathcal{T}^{k}} \mathbf{v} \|_{2} }{\sqrt{ K - \ell'}},      \label{eq:ggeeewwwq}
    \end{eqnarray}
    where (a) is due to Cauchy-Schwarz inequality and (b) is from the triangle inequality.
Observe that
\begin{eqnarray}
 \left\| \mathbf{\Phi}'_{\mathcal{T}  }
      \mathbf{P}^\bot_{\mathcal{T}^{k}} \mathbf{v} \right\|_{2}
 & = & \left\| (\mathbf{P}^\bot_{\mathcal{T}^{k}} \mathbf{\Phi}_{\mathcal{T}  })'
       \mathbf{v} \right\|_{2} \nonumber \\
 & {\leq} & \sqrt{\lambda_{\max} \left(  (\mathbf{P}^\bot_{\mathcal{T}^{k}} \mathbf{\Phi}_{\mathcal{T}  })' \mathbf{P}^\bot_{\mathcal{T}^{k}} \mathbf{\Phi}_{\mathcal{T}  } \right)} \left\| \mathbf{v} \right\|_{2} \nonumber \\
 & \overset{\eqref{eq:Pbot}}{=} & \sqrt{\lambda_{\max} \left(  \mathbf{\Phi}'_{\mathcal{T}  } \mathbf{P}^\bot_{\mathcal{T}^{k}} \mathbf{\Phi}_{\mathcal{T}  } \right)} \left\| \mathbf{v} \right\|_{2} \nonumber \\
 & \overset{\text{Lemma}~\ref{lem:rip6}}{\leq} & \sqrt{\lambda_{\max} \left( \mathbf{\Phi}'_{\mathcal{T} \cup \mathcal{T}^k  }  \mathbf{\Phi}_{\mathcal{T} \cup \mathcal{T}^k } \right)} \left\| \mathbf{v} \right\|_{2} \nonumber \\
 & {\leq} & \sqrt{1 + \delta_{K + Lk - \ell'}} \left\| \mathbf{v} \right\|_{2}.  \label{eq:969}
\end{eqnarray}
Also, from \eqref{eq:geaig1aa1}, we have
\begin{equation}
    \| \mathbf{\Phi}'_{\mathcal{T}  \backslash \mathcal{T}^k } \mathbf{P}^\bot_{\mathcal{T}^k}  \mathbf{\Phi}  \mathbf{x} \|_2 \geq ( 1- \delta_{K + Lk -
      \ell'} ) \left\| \mathbf{x}_{\mathcal{T} \setminus \mathcal{T}^{k}}  \right\|_{2}.
      \label{eq:geaig1aa}
    \end{equation}

Using \eqref{eq:ggeeewwwq}, \eqref{eq:geaig1aa} and \eqref{eq:969}, we obtain {\eqref{eq:small5}}.

\vspace{1mm}

{\it 2) Proof of {\eqref{eq:large5}}}:
    Let $\mathcal{F}'$ be the index set corresponding to $L$ largest elements in $\big \{\frac{|
    \langle \phi_{i}, \mathbf{r}^{k} \rangle |}{\| \mathbf{P}^{\bot}_{\mathcal{T}^{k}}
    \phi_{i} \|_{2}} \big \}_{i \in \Omega \setminus (\mathcal{T} \cup \mathcal{T}^{k} )}$. Following \eqref{eq:8700}, we can show that
    \begin{equation}
    \left( \sum_{i \in \mathcal{F}'} \frac{| \langle \phi_{i},
      \mathbf{r}^{k} \rangle |^{2}}{\| \mathbf{P}^{\bot}_{\mathcal{T}^{k}} \phi_{i}
      \|_{2}^{2}} \right)^{\hspace{-1.2mm} 1/2} \hspace{-1.2mm}\leq \hspace{-.5mm}\left( {1- \frac{\delta_{Lk+1}^{2}}{1 \hspace{-.5mm}- \hspace{-.5mm}
      \delta_{Lk}}} \right)^{\hspace{-1.2mm}-1/2} \hspace{-1.2mm}\| \mathbf{\Phi}'_{\mathcal{F}'} \mathbf{r}^{k} \|_{2}. \label{eq:con0}
    \end{equation}
  Observe that
  \begin{eqnarray}
 \| \mathbf{\Phi}'_{\mathcal{F}'} \mathbf{r}^{k} \|_{2}
  & = & \| \mathbf{\Phi}'_{\mathcal{F}'} \mathbf{P}^\bot_{\mathcal{T}^{k}} (\mathbf{\Phi x} + \mathbf{v})\|_{2} \nonumber \\
  & \leq & \| \mathbf{\Phi}'_{\mathcal{F}'} \mathbf{P}^\bot_{\mathcal{T}^{k}} \mathbf{\Phi x} \|_2 + \| \mathbf{\Phi}'_{\mathcal{F}'} \mathbf{P}^\bot_{\mathcal{T}^{k}}  \mathbf{v}\|_{2}  \nonumber \\
  & = & \| \mathbf{\Phi}'_{\mathcal{F}'} \mathbf{P}^\bot_{\mathcal{T}^{k}} \mathbf{\Phi}_{\mathcal{T} \setminus
      \mathcal{T}^{k} } \mathbf{x}_{\mathcal{T} \setminus
      \mathcal{T}^{k} } \|_2 + \| \mathbf{\Phi}'_{\mathcal{F}'} \mathbf{P}^\bot_{\mathcal{T}^{k}}  \mathbf{v}\|_{2}  \nonumber \\
  & \leq & \| \mathbf{\Phi}'_{\mathcal{F}'}  \mathbf{\Phi}_{\mathcal{T} \setminus
      \mathcal{T}^{k} } \mathbf{x}_{\mathcal{T} \setminus
      \mathcal{T}^{k} } \|_2 + \| \mathbf{\Phi}'_{\mathcal{F}'} \mathbf{P}^\bot_{\mathcal{T}^{k}}  \mathbf{v}\|_{2}  \nonumber \\
  &&  +  \| \mathbf{\Phi}'_{\mathcal{F}'} \mathbf{P}_{\mathcal{T}^{k}} \mathbf{\Phi}_{\mathcal{T} \setminus
      \mathcal{T}^{k} } \mathbf{x}_{\mathcal{T} \setminus
      \mathcal{T}^{k} } \|_2  \label{eq:10009}
  \end{eqnarray}
  Following (\ref{eq:j1}) and (\ref{eq:ghg2}), we have
\begin{eqnarray}
         \left\| \mathbf{\Phi}'_{\mathcal{F}'}  \mathbf{\Phi}_{\mathcal{T} \setminus
      \mathcal{T}^{k} }
      \mathbf{x}_{\mathcal{T} \setminus
      \mathcal{T}^{k} } \right\|_{2} \hspace{-3mm}&\leq& \hspace{-3mm} \delta_{L+K- \ell'}
      \| \mathbf{x}_{\mathcal{T} \setminus
      \mathcal{T}^{k}} \|_{2}, \label{eq:j15}
\\
  \left\| \mathbf{\Phi}'_{\mathcal{F}'} \mathbf{P}_{\mathcal{T}^{k}} \mathbf{\Phi}_{\mathcal{T} \setminus \mathcal{T}^{k}} \mathbf{x}_{\mathcal{T} \setminus
      \mathcal{T}^{k} } \hspace{-.51mm} \right\|_{2} \hspace{-3mm}&\leq& \hspace{-3.5mm}\frac{\delta_{L+Lk} \delta_{Lk+K- \ell'} \hspace{-.75mm} \left\| \mathbf{x}_{\mathcal{T} \setminus
      \mathcal{T}^{k} } \hspace{-.51mm} \right\|_{2}}{1-
      \delta_{Lk}}\hspace{-.5mm}.~~~~~~ \label{eq:ghg255}
\end{eqnarray}
Also,
\begin{eqnarray}
\| \mathbf{\Phi}'_{\mathcal{F}'} \mathbf{P}^\bot_{\mathcal{T}^{k}}  \mathbf{v}\|_{2}
\hspace{-3mm}&=&\hspace{-3mm} \| (\mathbf{P}^\bot_{\mathcal{T}^{k}} \mathbf{\Phi}_{\mathcal{F}'})'   \mathbf{v}\|_{2}  \nonumber \\
\hspace{-3mm}&\leq& \hspace{-3mm}\sqrt{\lambda_{\max} \left( (\mathbf{P}^\bot_{\mathcal{T}^{k}} \mathbf{\Phi}_{\mathcal{F}'})' \mathbf{P}^\bot_{\mathcal{T}^{k}} \mathbf{\Phi}_{\mathcal{F}'} \right)} \|\mathbf{v} \|_2 \nonumber \\
\hspace{-3mm} & \overset{\eqref{eq:Pbot}}{=} & \hspace{-3mm}\sqrt{\lambda_{\max} \left(   \mathbf{\Phi}'_{\mathcal{F}'} \mathbf{P}^\bot_{\mathcal{T}^{k}} \mathbf{\Phi}_{\mathcal{F}'} \right)} \|\mathbf{v} \|_2 \nonumber \\
\hspace{-3mm}&\overset{\text{Lemma}~\ref{lem:rip6}}{\leq}& \hspace{-3mm}\sqrt{\lambda_{\max} \left(   \mathbf{\Phi}'_{\mathcal{F}' \cup \mathcal{T}^{k}} \mathbf{\Phi}_{\mathcal{F}' \cup \mathcal{T}^{k}} \right)} \|\mathbf{v} \|_2 \nonumber \\
\hspace{-3mm}&\leq& \hspace{-3mm}\sqrt{1 \hspace{-.5mm} + \hspace{-.5mm} \delta_{|\mathcal{F}' \cup \mathcal{T}^k|}} \|\mathbf{v} \|_2 \hspace{-.5mm} = \hspace{-1mm} \sqrt{1 \hspace{-.5mm} + \hspace{-.5mm} \delta_{L + Lk}} \|\mathbf{v} \|_2. \nonumber \\ \label{eq:ghg2553}
\end{eqnarray}
    Using (\ref{eq:con0}), \eqref{eq:10009}, (\ref{eq:j15}), (\ref{eq:ghg255}), and \eqref{eq:ghg2553}, we have
    \begin{eqnarray}
      \lefteqn{\left( \sum_{i \in \mathcal{F}'} \frac{| \langle \phi_{i}, \mathbf{r}^{k} \rangle
      |^{2}}{\| \mathbf{P}^{\bot}_{\mathcal{T}^{k}} \phi_{i} \|_{2}^{2}}
      \right)^{\hspace{-1mm}1/2} \hspace{-1mm} \leq \left( \frac{1- \delta_{Lk}}{1- \delta_{Lk} - \delta_{Lk+1}^{2}}
      \right)^{\hspace{-1mm} 1/2} \left( {\left\| \mathbf{x} _{\mathcal{T} \setminus
      \mathcal{T}^{k} } \right\|_{2}} \right. } \nonumber \\
      & & \hspace{-6mm}\times \left.\hspace{-.5mm} \left(\hspace{-.5mm} \delta_{L+K- \ell'} \hspace{-.5mm}+\hspace{-.5mm}
      \frac{\delta_{L+Lk} \delta_{Lk+K- \ell'}}{1 - \delta_{Lk}} \hspace{-.5mm} \right)
       \hspace{-.5mm} + \hspace{-.5mm} \sqrt{1 \hspace{-.5mm} + \hspace{-.5mm}\delta_{L + Lk}} \|\mathbf{v}\|_2 \hspace{-.5mm} \right)\hspace{-.5mm}.~~~~~~   \label{eq:left5}
    \end{eqnarray}
    On the other hand, by noting that $v'_{L}$ is the $L$-th largest value in
    $\big \{ \frac{| \langle \phi_{i}, \mathbf{r}^{k} \rangle |}{\| \mathbf{P}^{\bot}_{\mathcal{T}^{k}} \phi_{i} \|_{2}} \big \}_{i \in \mathcal{F}'}$, we have
    \begin{eqnarray}
      \left( \sum_{i \in \mathcal{F}'} \frac{| \langle \phi_{i}, \mathbf{r}^{k} \rangle
      |^{2}}{\| \mathbf{P}^{\bot}_{\mathcal{T}^{k}} \phi_{i} \|_{2}^{2}} \right)^{1/2}
      \geq \sqrt{L} v'_{L}.  \label{eq:right005}
    \end{eqnarray}

  Using {\eqref{eq:left5}} and {\eqref{eq:right005}}, we obtain {\eqref{eq:large5}}.
\end{proof}

\section{Proof of \eqref{eq:k+15}} \label{app:cond5}

\begin{proof}
Rearranging the terms in \eqref{eq:48o} we obtain
  \begin{eqnarray}
\hspace{-8mm} &&\sqrt{\frac{L}{K - \ell'}} (1 - \delta_{LK}) - \frac{\delta_{LK}}{1 - \delta_{LK}} \left(1 +  \frac{\delta_{LK}^2}{1- \delta_{LK}
    - \delta_{LK}^{2}} \right)^{1/2}  \nonumber \\
 \hspace{-8mm}   && > \hspace{-.5mm} \left(\hspace{-1mm} \left(\hspace{-1mm} 1 \hspace{-.5mm} + \hspace{-.5mm} \frac{\delta_{LK}^2}{1 \hspace{-.5mm} - \hspace{-.5mm} \delta_{LK} \hspace{-.5mm}
    - \hspace{-.5mm} \delta_{LK}^{2}} \right)^{\hspace{-1mm} 1/2} \hspace{-2mm} + \hspace{-1mm} \sqrt{\frac{L}{K \hspace{-.5mm} - \hspace{-.5mm} \ell'}} \right) \hspace{-1mm} \frac{\sqrt{1 \hspace{-.5mm} + \hspace{-.5mm} \delta_{LK}} \|\mathbf{v}\|_2}{\|\mathbf{x}_{\mathcal{T} \setminus \mathcal{T}^{k}} \|_2}. \label{eq:110001}
  \end{eqnarray}
In the following, we will show that \eqref{eq:110001} is guaranteed by~\eqref{eq:k+15}.
First, since
\begin{eqnarray}
\hspace{-.5mm} \|\mathbf{x}_{\mathcal{T} \backslash \mathcal{T}^k} \|_2 \hspace{-1mm}&\geq& \hspace{-1mm} \sqrt{|\mathcal{T} \backslash \mathcal{T}^k \hspace{-.5mm}|} \min_{j \in \mathcal{T}} |x_j| ~ \overset{\eqref{eq:snrmar}}{=} ~\frac{\kappa \sqrt{K - \ell'}  \|\mathbf{x}\|_2}{\sqrt K} \nonumber \\
\hspace{-1mm} &\overset{\text{RIP}}{\geq}& \hspace{-1mm} \frac{\kappa \sqrt{K \hspace{-.75mm} - \hspace{-.5mm} \ell'}  \|\mathbf{\Phi x}\|_2}{\sqrt{K (1 + \delta_{LK})}} ~~~\hspace{-.5mm} \overset{\eqref{eq:snrmar}}{=} ~ \frac{\kappa \sqrt{(K \hspace{-.75mm} - \hspace{-.5mm} \ell') snr}\|\mathbf{v}\|_2}{\sqrt{K ( 1 + \delta_{LK})}}, \nonumber
\end{eqnarray}
by denoting
\begin{eqnarray}
\beta &:=& \left(1 +  \frac{\delta_{LK}^2}{1- \delta_{LK}
    - \delta_{LK}^{2}} \right)^{1/2}, \nonumber \\
    \gamma &:=& \sqrt{\frac{L}{K - \ell'}}, \nonumber \\
    \delta &:=& \delta_{LK}, \nonumber \\
    \tau &:=& \frac{\sqrt K}{\kappa \sqrt{(K - \ell') snr}}, \nonumber
\end{eqnarray}
we can rewrite \eqref{eq:110001} as
\begin{equation}
\gamma (1 - \delta - ( 1 + \delta) \tau) > \beta \left( \frac{\delta}{1 - \delta} + (1 + \delta) \tau \right). \label{eq:118a}
\end{equation}
Since \eqref{eq:fggg} implies $
  \beta < \frac{(1 - \delta)^2}{1 - 2 \delta}$, it is easily shown that \eqref{eq:110001} is ensured by
\begin{equation}
\gamma > \frac{1 - \delta}{1 - 2 \delta} \cdot \frac{你\delta + (1 - \delta^2) \tau}{1 - \delta - (1 + \delta) \tau}. \label{eq:H3}
\end{equation}
Moreover, since  $1 - \delta^2 \leq 1 + \delta$, \eqref{eq:H3} holds true under
$
\gamma > \frac{1 - \delta}{1 - 2 \delta} \cdot \frac{你\delta + (1 + \delta) \tau}{1 - \delta - (1 + \delta) \tau},$
or equivalently,
\begin{equation}
\delta < \frac{1}{1 + \tau} \left( \frac{\gamma}{u + \gamma} - \tau\right)~\text{where}~u := \frac{1 - \delta}{1 - 2 \delta}. \label{eq:guuuu}
\end{equation}

Next, observe that
\begin{eqnarray}
\delta < \frac{\sqrt L}{\sqrt{K - \ell'} + 2 \sqrt L} &\Leftrightarrow& \delta < \frac{\gamma}{1 + 2 \gamma}, \nonumber \\
&\Leftrightarrow& \frac{1 - \delta}{1 - 2 \delta} < 1 + \gamma, \nonumber \\
&\Leftrightarrow& u < 1 + \gamma, \nonumber \\
&\Leftrightarrow& \frac{\gamma}{u + \gamma} > \frac{\gamma}{1 + 2 \gamma}. \label{eq:h6}
\end{eqnarray}
Thus, if $\delta < \frac{\sqrt L}{\sqrt{K - \ell'} + 2 \sqrt L}$, then we can derive from \eqref{eq:guuuu} and~\eqref{eq:h6} that \eqref{eq:110001} holds true whenever
\begin{equation}
\delta \hspace{-.25mm} < \hspace{-.25mm} \frac{1}{1 \hspace{-.5mm} + \hspace{-.5mm} \tau} \hspace{-.5mm} \left(\hspace{-.5mm} \frac{\gamma}{1 \hspace{-.5mm} + \hspace{-.5mm} 2\gamma} \hspace{-.5mm} - \hspace{-.5mm} \tau \hspace{-.5mm} \right) \hspace{-.75mm} = \hspace{-.5mm} \frac{1}{1 \hspace{-.5mm} + \hspace{-.5mm} \tau} \hspace{-.75mm} \left( \hspace{-.5mm}\frac{\sqrt L}{\sqrt{K \hspace{-.5mm} - \hspace{-.5mm} \ell'}  \hspace{-.5mm} +  \hspace{-.5mm} 2 \sqrt{L}} \hspace{-.5mm} - \hspace{-.5mm} \tau \hspace{-.5mm} \right)\hspace{-1mm}.  \label{eq:alsr}
\end{equation}
That is
\begin{equation}
\hspace{-1.5mm} \sqrt{snr} > \hspace{-.5mm} \frac{(1 + \delta_{LK}) (\sqrt{K - \ell'} + \sqrt L)}{\kappa (\sqrt L \hspace{-.5mm} - \hspace{-.5mm} (\sqrt{K \hspace{-.5mm} - \ell'} \hspace{-.5mm} + \hspace{-.5mm} 2 \sqrt L) \delta_{LK}) \sqrt{K \hspace{-.5mm} - \ell'}} \sqrt K. \label{eq:126o}
 \end{equation}

Finally, since $K - \ell' < K$, \eqref{eq:126o} is guaranteed by \eqref{eq:k+15},  this completes the proof.
  \end{proof}

\bibliographystyle{IEEEbib}
\bibliography{CS_refs}

\begin{thebibliography}{10}

\bibitem{donoho1989uncertainty}
D.~L. Donoho and P.~B. Stark,
\newblock ``{Uncertainty principles and signal recovery},''
\newblock {\em SIAM Journal on Applied Mathematics}, vol. 49, no. 3, pp.
  906--931, 1989.

\bibitem{donoho2006compressed}
D.~L. Donoho,
\newblock ``Compressed sensing,''
\newblock {\em IEEE Trans. Inform. Theory}, vol. 52, no. 4, pp. 1289--1306,
  Apr. 2006.

\bibitem{candes2006near}
E.~J. Cand{\`e}s and T.~Tao,
\newblock ``{Near-optimal signal recovery from random projections: Universal
  encoding strategies?},''
\newblock {\em IEEE Trans. Inform. Theory}, vol. 52, no. 12, pp. 5406--5425,
  Dec. 2006.

\bibitem{candes2006robust}
E.~J. Cand{\`e}s, J.~Romberg, and T.~Tao,
\newblock ``{Robust uncertainty principles: Exact signal reconstruction from
  highly incomplete frequency information},''
\newblock {\em IEEE Trans. Inform. Theory}, vol. 52, no. 2, pp. 489--509, Feb.
  2006.

\bibitem{chen2001atomic}
S.~S. Chen, D.~L. Donoho, and M.~A. Saunders,
\newblock ``Atomic decomposition by basis pursuit,''
\newblock {\em SIAM review}, pp. 129--159, 2001.

\bibitem{pati1993orthogonal}
Y.~C. Pati, R.~Rezaiifar, and P.~S. Krishnaprasad,
\newblock ``Orthogonal matching pursuit: Recursive function approximation with
  applications to wavelet decomposition,''
\newblock in {\em Proc. 27th Annu. Asilomar Conf. Signals, Systems, and
  Computers}. IEEE, Nov. Pacific Grove, CA, Nov. 1993, vol.~1, pp. 40--44.

\bibitem{mallat1993matching}
S.~G. Mallat and Z.~Zhang,
\newblock ``{Matching pursuits with time-frequency dictionaries},''
\newblock {\em IEEE Trans. Signal Process.}, vol. 41, no. 12, pp. 3397--3415,
  Dec. 1993.

\bibitem{chen1989orthogonal}
S.~Chen, S.~A. Billings, and W.~Luo,
\newblock ``Orthogonal least squares methods and their application to
  non-linear system identification,''
\newblock {\em International Journal of control}, vol. 50, no. 5, pp.
  1873--1896, 1989.

\bibitem{donoho2006sparse}
D.~L. Donoho, I.~Drori, Y.~Tsaig, and J.~L. Starck,
\newblock ``{Sparse solution of underdetermined linear equations by stagewise
  orthogonal matching pursuit},''
\newblock {\em IEEE Trans. Inform. Theory}, vol. 58, no. 2, pp. 1094--1121,
  Feb. 2012.

\bibitem{needell2010signal}
D.~Needell and R.~Vershynin,
\newblock ``{Signal recovery from incomplete and inaccurate measurements via
  regularized orthogonal matching pursuit},''
\newblock {\em IEEE J. Sel. Topics Signal Process.}, vol. 4, no. 2, pp.
  310--316, Apr. 2010.

\bibitem{wang2012Generalized}
J.~Wang, S.~Kwon, and B.~Shim,
\newblock ``Generalized orthogonal matching pursuit,''
\newblock {\em IEEE Trans. Signal Process.}, vol. 60, no. 12, pp. 6202--6216,
  Dec. 2012.

\bibitem{needell2009cosamp}
D.~Needell and J.~A. Tropp,
\newblock ``Cosamp: Iterative signal recovery from incomplete and inaccurate
  samples,''
\newblock {\em Applied and Computational Harmonic Analysis}, vol. 26, no. 3,
  pp. 301--321, Mar. 2009.

\bibitem{foucart2011hard}
S.~Foucart,
\newblock ``Hard thresholding pursuit: an algorithm for compressive sensing,''
\newblock {\em SIAM Journal on Numerical Analysis}, vol. 49, no. 6, pp.
  2543--2563, 2011.

\bibitem{dai2009subspace}
W.~Dai and O.~Milenkovic,
\newblock ``{Subspace pursuit for compressive sensing signal reconstruction},''
\newblock {\em IEEE Trans. Inform. Theory}, vol. 55, no. 5, pp. 2230--2249,
  May. 2009.

\bibitem{chartrand2007exact}
R.~Chartrand,
\newblock ``Exact reconstruction of sparse signals via nonconvex
  minimization,''
\newblock {\em IEEE Signal Process. Lett.}, vol. 14, no. 10, pp. 707--710, Oct.
  2007.

\bibitem{chartrand2008iteratively}
R.~Chartrand and W.~Yin,
\newblock ``Iteratively reweighted algorithms for compressive sensing,''
\newblock in {\em Proc. Int. Conf. Acoust., Speech, Signal Process. (ICASSP)}.
  IEEE, 2008, pp. 3869--3872.

\bibitem{foucart2009sparsest}
S.~Foucart and M~Lai,
\newblock ``Sparsest solutions of underdetermined linear systems via
  ℓq-minimization for $0 < q \leq 1$,''
\newblock {\em Applied and Computational Harmonic Analysis}, vol. 26, no. 3,
  pp. 395--407, 2009.

\bibitem{daubechies2010iteratively}
I.~Daubechies, R.~DeVore, M.~Fornasier, and C.~S. G{\"u}nt{\"u}rk,
\newblock ``Iteratively reweighted least squares minimization for sparse
  recovery,''
\newblock {\em Communications on Pure and Applied Mathematics}, vol. 63, no. 1,
  pp. 1--38, 2010.

\bibitem{chen2014convergence}
L.~Chen and Y.~Gu,
\newblock ``The convergence guarantees of a non-convex approach for sparse
  recovery using regularized least squares,''
\newblock in {\em Proc. Int. Conf. Acoust., Speech, Signal Process. (ICASSP)}.
  IEEE, 2014, pp. 3350--3354.

\bibitem{tropp2004greed}
J.~A. Tropp,
\newblock ``{Greed is good: Algorithmic results for sparse approximation},''
\newblock {\em IEEE Trans. Inform. Theory}, vol. 50, no. 10, pp. 2231--2242,
  Oct. 2004.

\bibitem{tropp2007signal}
J.~A. Tropp and A.~C. Gilbert,
\newblock ``{Signal recovery from random measurements via orthogonal matching
  pursuit},''
\newblock {\em IEEE Trans. Inform. Theory}, vol. 53, no. 12, pp. 4655--4666,
  Dec. 2007.

\bibitem{davenport2010analysis}
M.~A. Davenport and M.~B. Wakin,
\newblock ``{Analysis of Orthogonal Matching Pursuit using the restricted
  isometry property},''
\newblock {\em IEEE Trans. Inform. Theory}, vol. 56, no. 9, pp. 4395--4401,
  Sep. 2010.

\bibitem{zhang2011sparse}
T.~Zhang,
\newblock ``Sparse recovery with orthogonal matching pursuit under rip,''
\newblock {\em IEEE Trans. Inform. Theory}, vol. 57, no. 9, pp. 6215--6221,
  Sep. 2011.

\bibitem{mo2012remarks}
Q.~Mo and Y.~Shen,
\newblock ``A remark on the restricted isometry property in orthogonal matching
  pursuit algorithm,''
\newblock {\em IEEE Trans. Inform. Theory}, vol. 58, no. 6, pp. 3654--3656,
  Jun. 2012.

\bibitem{wang2012Recovery}
J.~Wang and B.~Shim,
\newblock ``On the recovery limit of sparse signals using orthogonal matching
  pursuit,''
\newblock {\em IEEE Trans. Signal Process.}, vol. 60, no. 9, pp. 4973--4976,
  Sep. 2012.

\bibitem{wen2013improved}
J.~Wen, X.~Zhu, and D.~Li,
\newblock ``Improved bounds on restricted isometry constant for orthogonal
  matching pursuit,''
\newblock {\em Electronics Letters}, vol. 49, no. 23, pp. 1487--1489, 2013.

\bibitem{wang2015support}
J.~Wang,
\newblock ``Support recovery with orthogonal matching pursuit in the presence
  of noise,''
\newblock {\em IEEE Trans. Signal Process.}, vol. 63, no. 21, pp. 5868--5877,
  Nov. 2015.

\bibitem{rebollo2002optimized}
L.~Rebollo-Neira and D.~Lowe,
\newblock ``Optimized orthogonal matching pursuit approach,''
\newblock {\em IEEE Signal Processing Letters}, vol. 9, no. 4, pp. 137--140,
  Apr. 2002.

\bibitem{foucart2013stability}
S.~Foucart,
\newblock ``Stability and robustness of weak orthogonal matching pursuits,''
\newblock in {\em Recent Advances in Harmonic Analysis and Applications}, pp.
  395--405. Springer, 2013.

\bibitem{soussen2013joint}
C.~Soussen, R.~Gribonval, J.~Idier, and C.~Herzet,
\newblock ``Joint $k$-step analysis of orthogonal matching pursuit and
  orthogonal least squares,''
\newblock {\em IEEE Trans. Inform. Theory}, vol. 59, no. 5, pp. 3158--3174, May
  2013.

\bibitem{liu2012orthogonal}
E.~Liu and V.~N. Temlyakov,
\newblock ``The orthogonal super greedy algorithm and applications in
  compressed sensing,''
\newblock {\em IEEE Trans. Inform. Theory}, vol. 58, no. 4, pp. 2040--2047,
  Apr. 2012.

\bibitem{candes2005decoding}
E.~J. Cand{\`e}s and T.~Tao,
\newblock ``{Decoding by linear programming},''
\newblock {\em IEEE Trans. Inform. Theory}, vol. 51, no. 12, pp. 4203--4215,
  Dec. 2005.

\bibitem{blumensath2007difference}
T.~Blumensath and M.~E. Davies,
\newblock ``On the difference between orthogonal matching pursuit and
  orthogonal least squares,''
\newblock 2007.

\bibitem{kwon2013multipath}
S.~Kwon, J.~Wang, and B.~Shim,
\newblock ``Multipath matching pursuit,''
\newblock {\em IEEE Trans. Inform. Theory}, vol. 60, no. 5, pp. 2986--3001, May
  2014.

\bibitem{candes2008restricted}
E.~J. Cand{\`e}s,
\newblock ``{The restricted isometry property and its implications for
  compressed sensing},''
\newblock {\em Comptes Rendus Mathematique}, vol. 346, no. 9-10, pp. 589--592,
  2008.

\bibitem{cai2011orthogonal}
T.~T. Cai and L.~Wang,
\newblock ``Orthogonal matching pursuit for sparse signal recovery with
  noise,''
\newblock {\em IEEE Trans. Inform. Theory}, vol. 57, no. 7, pp. 4680--4688,
  Jul. 2011.

\bibitem{chang2014improved}
L~Chang and J~Wu,
\newblock ``An improved \text{RIP}-based performance guarantee for sparse
  signal recovery via orthogonal matching pursuit,''
\newblock {\em IEEE Trans. Inform. Theory}, vol. 60, no. 9, pp. 5702--5715,
  Sep. 2014.

\bibitem{chang2013achievable}
L.~Chang and J.~Wu,
\newblock ``Achievable angles between two compressed sparse vectors under
  norm/distance constraints imposed by the restricted isometry property: a
  plane geometry approach,''
\newblock {\em IEEE Trans. Inform. Theory}, vol. 59, no. 4, pp. 2059--2081,
  Apt. 2013.

\bibitem{herzet2012exact}
C.~Herzet, C.~Soussen, J.~Idier, and R.~Gribonval,
\newblock ``Exact recovery conditions for sparse representations with partial
  support information,''
\newblock {\em IEEE Trans. Inform. Theory}, vol. 59, no. 11, pp. 7509--7524,
  Nov 2013.

\bibitem{fletcher2012orthogonal}
A.~K. Fletcher and S.~Rangan,
\newblock ``Orthogonal matching pursuit: A brownian motion analysis,''
\newblock {\em IEEE Trans. Signal Process.}, vol. 60, no. 3, pp. 1010--1021,
  March 2012.

\bibitem{wu2013exact}
R.~Wu, W.~Huang, and D~Chen,
\newblock ``The exact support recovery of sparse signals with noise via
  orthogonal matching pursuit,''
\newblock {\em IEEE Signal Processing Letters}, vol. 20, no. 4, pp. 403--406,
  April 2013.

\bibitem{candes2005error}
E.~Candes, M.~Rudelson, T.~Tao, and R.~Vershynin,
\newblock ``Error correction via linear programming,''
\newblock in {\em in IEEE Symposium on Foundations of Computer Science
  (FOCS).}, 2005, pp. 668--681.

\bibitem{nesterov1994interior}
Y.~Nesterov and A.~Nemirovskii,
\newblock {\em Interior-point polynomial algorithms in convex programming},
\newblock SIAM, 1994.

\end{thebibliography}

\end{document}